\documentclass[11pt]{article}

\usepackage[utf8]{inputenc}
\usepackage[english]{babel}
\usepackage[]{natbib}%
\usepackage{amsthm,amsmath,mathtools,amsfonts,amssymb}
\usepackage{bm}
\usepackage{multirow}

\usepackage[title]{appendix}

\usepackage[table]{xcolor}

\usepackage{graphicx}
\graphicspath{{./figures}}

\usepackage{caption}
\usepackage{subcaption}

\usepackage{comment}

\usepackage{array}

\usepackage[ruled]{algorithm}
\usepackage{algpseudocode}
\algnewcommand{\lIf}[1]{\State\algorithmicif\ #1\ \algorithmicthen}
\algnewcommand{\EndlIf}{\unskip\ \algorithmicend\ \algorithmicif}

\theoremstyle{plain}

\newtheorem{theorem}{Theorem}
\newtheorem{lemma}[theorem]{Lemma}
\newtheorem{corollary}[theorem]{Corollary}
\newtheorem{prop}[theorem]{Proposition}

\theoremstyle{definition}

\newcommand{\vctr}[1]{\bm{#1}} 

\newcommand{\N}{\mathbb{N}}
\newcommand{\R}{\mathbb{R}}

\newcommand{\G}{\mathcal{G}}

\newcommand{\bnd}[1]{\mathrm{bd}(#1)}

\newcommand{\sigmavh}{\hat{\vctr{\sigma}}}

\newcommand{\Sigmaii}{\Sigma_{-i,-i}}
\newcommand{\Sigmaiinv}{\left(\Sigma_{-i,-i}\right)^{-1}}
\newcommand{\sigmaii}{\vctr{\sigma}_{-i,i}} 
\newcommand{\sigmabdi}{\vctr{\sigma}_{\bnd{i},i}}
\newcommand{\sigmabdiest}{\hat{\vctr{\sigma}}_{\bnd{i},i}}
\newcommand{\sii}{\vctr{s}_{-i,i}} 
\newcommand{\Sii}{S_{-i,-i}}

\DeclareMathOperator*{\argmax}{arg\,max}
\DeclareMathOperator*{\diagOP}{diag}

\DeclareMathOperator{\tr}{tr}

\newcommand{\diag}[1]{\diagOP\left(#1\right)}
\newcommand{\trace}[1]{\tr\left[#1\right]}

\newcommand{\Ps}{\mathcal{P}}
\newcommand{\Hs}{\mathcal{H}}
\newcommand{\Set}{\mathcal{S}}

\newcommand{\algonamefull}{Generalised Iterative Conditional Fitting}
\newcommand{\algoname}{GICF}
\newcommand{\packagename}{\texttt{gicf}}

\usepackage{hyperref}

\newcounter{rowcntr}[table]
\newcolumntype{N}{>{\refstepcounter{rowcntr}}c}

\begin{document}

\title{A unified approach to penalized likelihood estimation of covariance matrices in high dimensions}

\author{
Luca Cibinel\\ University of Oslo\\
\texttt{lucaci@math.uio.no}
\and
Alberto Roverato\\ University of Padova\\
\texttt{alberto.roverato@unipd.it}
\and
Veronica Vinciotti\\ University of Trento\\
\texttt{veronica.vinciotti@unitn.it}
}

\date{}

\maketitle

\begin{abstract}
We consider the problem of estimation of a covariance matrix for Gaussian data in a high dimensional setting. Existing approaches  include maximum likelihood estimation under a pre-specified sparsity pattern, $\ell_{1}$-penalized loglikelihood optimization and ridge regularization of the sample covariance. We show that these three approaches can be addressed in an unified way, by considering the constrained optimization of an objective function that involves two suitably defined penalty terms. This unified procedure exploits the advantages of each individual approach, while bringing novelty in the combination of the three.
We provide an efficient algorithm for the optimization of the regularized objective function and describe the relationship between the two penalty terms, thereby highlighting the importance of the joint application of the three methods. A simulation study shows how the sparse estimates of covariance matrices returned by the procedure are stable and accurate, both in low and high dimensional settings, and how their calculation is more efficient than existing approaches under a partially known sparsity pattern.
An illustration on sonar data shows is presented for the identification of the covariance structure among signals bounced off a certain material. The method is implemented in the publicly available \texttt{R} package \packagename.
\end{abstract}

\section{Introduction}
The task of estimating covariance matrices from randomly sampled data is of fundamental importance for a wide range of statistical methodologies. The usual sample covariance matrix performs poorly in the high dimensional setting, where the number of cases, $n$, is small compared with the number of features, $p$. Its performance degrades as the ratio $p/n$ increases, with the estimator becoming singular when $p\geq n$. One of the main difficulties in the estimation of large covariance matrices comes from the fact that the number of free parameters, $p(p+1)/2$, grows quadratically in the number of features. As a consequence of this, most approaches rely on the idea of reducing dimensionality by imposing parsimony through assumptions on the size and structure of the effective parameters. A prominent approach to reduce the number of free parameters is to assume that the covariance matrix is sparse, in the sense that many of the off-diagonal entries are equal to zero; see \citet{fan2016overview,pourahmadi2011covariance,pourahmadi2013high,chi2014stable,lam2020high} for comprehensive overviews. A challenge for these approaches is being able to enforce the positive-definiteness constraint of the solution \citep{xue2012positive, pourahmadi2011covariance,rothman2012positive}.

When the features follow a multivariate Gaussian distribution, zero covariances correspond to marginal independence relationships. In the graphical modelling literature \citep{maathuis2018handbook}, this corresponds to the case of a bi-directed graph. Within this framework, \citet{chaudhuri2007estimation} considered the problem of the computation of the maximum likelihood estimate of the covariance matrix in the case where the zero pattern is known, and introduced an algorithm called the Iterative Conditional Fitting (ICF). The algorithm is efficient, as it decomposes the problem into a sequence of least squares optimization steps, but its
convergence has been proved only in the case when the sample covariance matrix has full rank.

In the area of Gaussian graphical models, where the interest is for modelling conditional independence relationships \citep{lauritzen1996graphical},  penalized likelihood methods \citep{yuan2007model} are well-established tools for the identification of sparsity pattern of the inverse covariance matrix in high dimensional problems. In contrast to the breadth of this literature on inverse covariance estimation
\citep{drton2017structure}, only a small number of studies has considered penalized methods for the estimation of covariance matrices \citep{lam2009sparsistency,bien2011sparse, wang2014coordinate,xu2022proximal}.  This discrepancy can be explained with the fact that the negative Gaussian loglikelihood is convex as a function of the inverse covariance matrix. This both ensures that minimizers are global optima and that fast algorithms, such as the graphical lasso, can be applied \citep{friedman2008sparse}. However, the negative loglikelihood is nonconvex as a function of the covariance matrix, and, typically, it has multiple modes, leading to significant computational difficulties in this case.

\citet{bien2011sparse} and \citet{wang2014coordinate} introduced two algorithms for the optimization of the $\ell_{1}$-penalized loglikelihood for covariance estimation, that we will shortly refer to as the covariance graphical lasso (covglasso) penalization. The more efficient algorithm of \citet{wang2014coordinate} develops a coordinate-descent approach, in the spirit of the graphical lasso algorithm \citep{friedman2008sparse}. Although the connection is not made explicit in the paper, the coordinate-descent approach of \citet{wang2014coordinate} relies on the same loglikelihood decomposition of \citet{chaudhuri2007estimation}. Similarly to that case, the application of the \citet{bien2011sparse} and \citet{wang2014coordinate}  algorithms is restricted to the case where the sample covariance has full rank.  When this is not the case, the authors suggest to apply the algorithm to a ridge-regularized version of the sample covariance \citep{le2020comment} but do not pursue this further. A similar suggestion is made by other related approaches \citep{xu2022proximal}. In a different context, \citet{warton2008penalized} considered the problem of the ridge regularization of the sample covariance matrix. In particular, \citet{warton2008penalized} showed how the ridge covariance estimator can be formally obtained from the optimization of a Gaussian loglikelihood, with the addition of a penalization proportional to the trace of the inverse covariance matrix.

We focus on the estimation of sparse covariance matrices in the high dimensional setting. The overall idea at the basis of this work is that the approach of \citet{chaudhuri2007estimation} to the computation of the maximum likelihood estimates under a given sparsity pattern, the algorithm of \citet{wang2014coordinate} for the
optimization of the $\ell_{1}$-penalized loglikelihood and the approach of \citet{warton2008penalized} to ridge regularization, can be addressed in a unified way within an appropriately defined regularized problem. In particular, the Gaussian loglikelihood function is optimized under a possibly partially known structure of zeros in the covariance and, more importantly, under
a modified covglasso regularization where the penalization on the diagonal entries is replaced with a penalization on the diagonal entries of the inverse covariance.

By unifying the three approaches, the procedure simultaneously exploits the advantages of each individual method as well as bringing novelty in the combination of the three. Firstly, the unification of the ICF algorithm of \citet{chaudhuri2007estimation} and that of \citet{wang2014coordinate} results in a computational tool that includes the two methods, and that is considerably more efficient than \citet{wang2014coordinate} in the case of sparse estimation under a partially known sparsity pattern, such as when a banded structure of zeros is assumed. Secondly, the replacement of the $\ell_{1}$-penalization on the diagonal of the covariance with a penalty term applied to the diagonal of the inverse covariance results in a further generalization of the covglasso method that now also incorporates a ridge regularization component in the spirit of \citet{warton2008penalized}. This has a number of advantages. In particular, connected with the literature of lasso regression methods with error variance estimation \citep{yu2019estimating}, a penalization of the diagonal of the inverse covariance instead of the covariance has the effect of penalizing small error variances, rather than large ones.
Moreover, from the practical viewpoint, the ridge regularization component makes the optimization problem well-defined also when $S$ is singular, and therefore extends covglasso also to the high dimensional setting. In the case where the covglasso penalty is set to zero, our approach amounts to applying the ridge regularization method of \citet{warton2008penalized} to the estimation of covariance matrices under a given sparsity pattern, thereby expanding the area of application of the ICF algorithm of \citet{chaudhuri2007estimation}.

Unifying the approaches under the same framework has a further advantage of allowing simultaneous identification of the optimal amount of ridge and covglasso regularization. A large amount of penalization results in a diagonal estimate of the covariance, and we provide the maximum theoretical values of both the ridge regularization parameter and the covglasso penalty term. This facilitates the practical implementation of the method, but also makes it evident that the values of the two terms should not be determined independently of each other because their maximum theoretical values are interconnected. We provide a Generalized Iterative Conditional Fitting (\algoname) algorithm for the optimization of the objective function and the selection of the optimal tuning parameters via cross-validation. This is implemented in the \texttt{R} package \packagename; see the `code and data availability' section for details.
A simulation study shows that the optimal value of ridge regularization depends both on the amount of covglasso penalization and on other features, such as the sample size and the underlying level of sparsity. Although the application of ridge regularization is required when $n$ is smaller than $p$, our results show that it is generally useful also in the case where the ratio $p/n$ is small, even if $n$ is larger than $p$, while its importance vanishes when $n$ is sufficiently large compared to $p$.
Finally, we illustrate an application of our approach on sonar data, in the context of a classification problem whose performance depends on the quality of the estimated covariance matrices.

The rest of this paper is organized as follows. In Section~\ref{SEC:covariance estimation} we provide the background on covariance estimation, as required for this paper. The proposed unified approach is presented in
Section~\ref{SEC:unifying.approach}. The \algoname\ algorithm is described in Section~\ref{SEC:hybrid.procedure} whereas a more detailed pseudocode can be found in Appendix~\ref{APP:pseudocodes}. The practical implementation of the method, including the derivation of the maximum theoretical values of the penalty parameters, is given in Section~\ref{SCT:model.selection}. The simulation study and the illustration on sonar data can be found in Sections~\ref{SCT:simulations} and \ref{SEC:application}, respectively. Finally, Section~\ref{SEC:conclusions} contains a brief discussion. Proofs are deferred to Appendices~\ref{AX:bounds} and \ref{AX:proofs}.

\section{Sparse covariance estimation}\label{SEC:covariance estimation}
Let $Y$ be an $n\times p$ data matrix whose rows are i.i.d. observations from a Gaussian random vector with zero mean and covariance matrix $\Sigma=(\sigma_{jk})_{j,k\in V}$, with $V=\{1,\ldots, p\}$. The Maximum Likelihood Estimator (MLE) of $\Sigma$ maximizes the loglikelihood,
\begin{align}\label{EQ:log.lik}
\ell_{Y,n}(\Sigma) = -\log{|\Sigma|} -\frac{1}{n}\tr(\Sigma^{-1}Y^{\top}Y)= -\log{|\Sigma|} -\tr(\Sigma^{-1}S),
\end{align}
over the cone of positive definite matrices $\mathcal{S}_{\succ}$, where $|\Sigma|$ denotes the determinant of $\Sigma$, $\tr(\cdot)$ the trace of a matrix and $S=\dfrac{1}{n}Y^{\top}Y$ is the matrix of sums of squares and product of the columns of $Y$. As it will become clear in Proposition~\ref{THM:shrinkage.transform}, it is useful for us to write explicitly that the loglikelihood depends on the sample size $n$. When the latter coincides with the number of rows of $Y$, we will simplify the notation to $\ell_{Y}(\Sigma)$.
We recall that there is no loss of generality in the assumption that the mean vector is known. In the more general setting where the mean vector needs to be estimated, one can center the columns of $Y$ to have zero mean. Then, $S$ is the sample covariance matrix and $\ell_{Y}(\Sigma)$ in (\ref{EQ:log.lik}) is a profile loglikelihood; see \citet[][Section~2.2]{chaudhuri2007estimation} for more details.

We consider the case where $\Sigma$ is sparse in the sense that some off-diagonal entries are equal to zero. We represent the sparsity pattern through  a graph $\G=(V, E)$ with vertex set $V$ and edge set $E\subset V\times V$ made of unordered pairs of vertices $\{j, k\}$ with $j\neq k$. Two vertices $j,k\in V$ are said to be adjacent when they are connected by an edge, that is $\{j,k\}\in E$ and, with a slight abuse of notation, we also write this as $\{j,k\}\in \G$. Furthermore, a graph $\G$ is complete if
$\{j,k\}\in\G$ for every pair of distinct vertices $j,k\in V$. In the graphical modelling literature, these edges are conventionally represented as bidirected lines \citep{richardson2003markov}. We say that $\Sigma$ is adapted to a graph $\G=(V, E)$ if, for every $j\neq k$, it holds that $\{j,k\}\notin \G$ implies $\sigma_{jk}=0$. We denote by $\mathcal{S}_{\succ}(\G)$ the cone of positive definite matrices which are adapted to $\G$, and we remark that a matrix $\Sigma\in \mathcal{S}_{\succ}(\G)$ may have a zero pattern that is sparser than that encoded by $\G$. Indeed, for every graph $\G^{\prime}=(V, E^{\prime})$  such that $E^{\prime}\subseteq E$, it holds that $\mathcal{S}_{\succ}(\G^{\prime})\subseteq \mathcal{S}_{\succ}(\G)$.

\subsection{Maximum likelihood estimate} \label{sec:MLE}
Maximizing (\ref{EQ:log.lik}) over $\mathcal{S}_{\succ}(\G)$, that is under the assumption that $\Sigma$ is adapted to $\G$, is a challenging task because the objective function is not concave and may have multiple local maxima. A solution to this problem is provided by \citet{chaudhuri2007estimation} by means of a procedure called the Iterative Conditional Fitting (ICF) algorithm.

In order to describe this algorithm, consider, for every $i\in V$, a partition of the matrices $\Sigma$, $S$ and $Y$ as follows
\begin{equation}\label{EQ:decomp.utils}
\Sigma =
\begin{pmatrix}
\Sigmaii & \sigmaii\\
\sigmaii^\top & \sigma_{ii}
\end{pmatrix},
\qquad
S =
\begin{pmatrix}
\Sii & \sii\\
\sii^\top & s_{ii}
\end{pmatrix},
\qquad
Y =
\begin{pmatrix}
Y_{-i} & \vctr{y}_{i}
\end{pmatrix},
\end{equation}
where the subscripts are subsets of $V$ that we use as row and column indexes.  Note also that $-i$ and $i$ are shorthand for $V\setminus \{i\}$ and $\{i\}$, respectively, and that we use bold lowercase symbols for vectors. From this, we define
\[Z_{-i} = Y_{-i}\Sigmaiinv,\]
which will act as a matrix of pseudo-predictors in a regression model of the $i$-th variable against the remaining variables.
In order to see this, for an arbitrary $i\in V$, we partition the full loglikelihood in (\ref{EQ:log.lik}) into the loglikelihood $\ell_{Y_{-i}}(\Sigmaii)$ of the marginal distribution of all variables excluding the $i$-th variable  and the loglikelihood $\ell_{\vctr{y}_{i}\mid Z_{-i}}(\tau_{i}, \sigmaii)$ of the conditional distribution of the $i$-th variable given the remaining variables, namely
\begin{equation}\label{EQ:log.lik.split}
\ell_{Y}(\Sigma) = \underbrace{
-\log|\Sigmaii| -\trace{\Sigmaiinv S_{-i,-i}}
}_{\ell_{Y_{-i}}(\Sigmaii)}
\quad
\underbrace{-\log{\tau_i} - \frac{1}{n\tau_i}\|\vctr{y}_i - Z_{-i}\sigmaii\|^2_2
 }_{\ell_{\vctr{y}_{i}\mid Z_{-i}}(\tau_{i}, \sigmaii)},
\end{equation}
with
\begin{equation}\label{EQ:tau}
\tau_i = \sigma_{ii} - \sigmaii^\top\Sigmaiinv\sigmaii.
\end{equation}
\citet{chaudhuri2007estimation} notice how $\ell_{\vctr{y}_{i}\mid Z_{-i}}(\tau_{i}, \sigmaii)$ in (\ref{EQ:log.lik.split}) is the loglikelihood of the Gaussian linear model, where $\vctr{y}_{i}$ is regressed on the pseudo-variables $Z_{-i}$, with vector of regression coefficients $\sigmaii$ and error variance $\tau_i$.

Knowledge of the graph $\G$ means that some off-diagonal entries of $\Sigma$ are known to be zero. In particular, if we denote by $\bnd{i}$ the boundary of $i\in V$, that is the set of vertices adjacent
to $i$ in $\G$, then the structure imposed by $\G$ implies that the entries of $\sigmaii$ are such that  $\sigma_{ij}=0$ for every $j\notin \bnd{i}$. The ICF algorithm maximizes $\ell_{Y}(\Sigma)$ numerically, by iteratively applying the decomposition (\ref{EQ:log.lik.split}) for $i\in V$. At every step,
it maximizes the conditional loglikelihood $\ell_{\vctr{y}_{i}\mid Z_{-i}}(\tau_{i}, \sigmaii)$, similarly to the computation of the MLEs. In particular, it amounts to the application of the usual least-squares formulae to the reduced system defined by the zero constraints on the regression coefficients $\sigmaii$, and then to the computation of the usual estimate of $\tau_{i}$ based on the residual sum of squares.

\subsection{Maximum penalized loglikelihood estimate}\label{SCT:maximum.penal.lik}
If a pre-specified zero pattern in $\Sigma$ is not provided, or only partially provided, this needs to be identified from the available data. To this end, one can apply a lasso penalization \citep{tibshirani1996regression}, so that many parameters will be shrinked to be exactly equal to zero. When the interest is for likelihood-based regularization of the covariance matrix, this amounts to maximizing
\[\ell_{Y}(\Sigma) - \lambda\,\|\Sigma\|_{1},\]
where $\|\Sigma\|_{1}=\sum_{j,k\in V}|\sigma_{jk}|$ denotes the $\ell_{1}$ norm of $\Sigma$ and $\lambda\geq 0$ is the shrinkage parameter. We will refer to this optimization task and to its solution as the covariance graphical lasso (covglasso) problem  and the covglasso estimator, respectively. Similar to the optimization of the loglikelihood, also the covglasso  penalized loglikelihood  may have multiple local maxima. Algorithms for its optimization are provided by \citet{bien2011sparse} and \citet{wang2014coordinate}. More specifically, \citet{bien2011sparse} approach the problem with a majorization-minimization algorithm, whereas \citet{wang2014coordinate} implements a coordinate-descent algorithm that, similarly to the ICF algorithm, exploits the loglikelihood decomposition in (\ref{EQ:log.lik.split}).

It is worth mentioning that the covglasso problem is feasible of a more general formulation where different shrinkage parameters $\lambda_{jk}$ are allowed for the different entries $\sigma_{jk}$ of $\Sigma$. This feature can be exploited in the case where the sparsity pattern is partially known. In practical terms, if the desired zero pattern is encoded by a graph $\G=(V, E)$, then for all $\{j,k\}\notin E$, with $j\neq k$, the values $\lambda_{jk}$ are set to a sufficiently large value to assure that the solution belongs to $\mathcal{S}_{\succ}(\G)$. The remaining $\lambda$-terms can be set to a  positive value, in which case the covglasso estimator may have a sparser structure than that encoded by $\G$. If, instead, they are all set to zero, then the covglasso estimator will coincide with the MLE of $\Sigma\in \mathcal{S}_{\succ}(\G)$.

When estimating a sparse covariance matrix, it is important to distinguish between the diagonal and off-diagonal elements, since the diagonal entries cannot be equal to zero. Connected to the previous point, this can be done by associating a penalty $\lambda\geq 0$ to the off-diagonal entries and a, possibly different, penalty $\kappa\geq 0$ to the diagonal. In this more general formulation, the covglasso penalized loglikelihood can be written as
\begin{align}\label{EQ:covglasso.penalty}
\ell_{Y}(\Sigma)-\lambda\,\|\Sigma-\diag{\Sigma}\|_{1} - \kappa\,\|\diag{\Sigma}\|_{1},
\end{align}
where $\diag{\cdot}$ is the diagonal matrix operator. A common choice is to set $\kappa$ equal to zero, thereby completely excluding the diagonal from the penalization \citep{bien2011sparse}. This  was in fact the first formulation of the covglasso estimator and the case in which theoretical properties were investigated \citep{lam2009sparsistency}.

Although the focus of this work is on $\ell_{1}$ loglikelihood regularization, it is worth mentioning that alternative likelihood-based approaches for estimating sparse covariance matrices are possible. In particular, \citet{xu2022proximal} have recently introduced a likelihood-based estimate of $\Sigma$ by means of a proximal distance algorithm.

\subsection{High dimensional setting}

The existence of a global maxima of (\ref{EQ:log.lik}) is guaranteed only if $S$ has full rank \citep{chaudhuri2007estimation}. With probability one, this happens if and only if $n\geq p$, and the inequality is strict ($n>p$) in the case where the columns of $Y$ are centered \citep{dykstra1970establishing}. The same is true for the covglasso problem \citep{bien2011sparse}. Indeed, convergence of both the ICF and the algorithms introduced by \citet{bien2011sparse} and \citet{wang2014coordinate} is shown under the assumption that $S$ has full rank.

While the ICF algorithm is designed for the computation of MLEs, and thus its application makes sense only when MLEs exist,  the covglasso is meant to be applied in the high dimensional setting. So the restriction of $S$ being full rank is quite a limitation in this case. In their paper,  \cite{bien2011sparse} suggest that, when $S$ is singular, it should be replaced by $(S+\kappa\, I_{p})$, where $I_{p}$ is the $p\times p$ identity matrix and  $\kappa$ has a small positive value. In a different context, the same suggestion is given by \citet{xu2022proximal} with respect to their proximal distance algorithm. However, this suggestion is not pursued by either of the two studies and no indication is given on how to choose the value of $\kappa$, or the potential impact of this choice on the result. In contrast to this, in a regression framework, \citet{warton2008penalized}
establishes a connection between the ridge-regularized estimator $(S+\kappa\, I_{p})$ of $\Sigma$ and a maximum penalized loglikelihood estimator of $\Sigma$, which is then exploited in a procedure for the selection of the amount of ridge-regularization $\kappa\geq 0$.

Following  \citet{warton2008penalized}, the approach that we propose in this paper naturally includes this regularization within the main optimization problem. This not only makes the method feasible in the high dimensional setting, but it also allows us to select the regularization level in a statistically principled way.

\section{A unifying approach}\label{SEC:unifying.approach}

The formulation in (\ref{EQ:covglasso.penalty}) is typically used with $\kappa=0$, as variances cannot be zero \citep{bien2011sparse}. Here, we will consider instead the following penalty term,

\begin{align}\label{EQ:shrinkage-covglasso.penalty}
\begin{split}
\Ps_{\lambda,\kappa}(\Sigma)
&=\lambda\,\|\Sigma-\diag{\Sigma}\|_{1} + \kappa\,\|\diag{\Sigma^{-1}}\|_{1}\\[1ex]
&=\lambda\,\sum_{j,k\in V; j\neq k} |\sigma_{jk}| + \kappa\,\sum_{j\in V}\theta_{jj}.
\end{split}
\end{align}
That is, we modify the covglasso penalty by replacing the diagonal of the covariance matrix with the diagonal of its inverse, $\Sigma^{-1}=(\theta_{jk})_{j,k\in V}$, also called the concentration matrix. Accordingly, we focus on the estimator resulting from the application of the penalty (\ref{EQ:shrinkage-covglasso.penalty}) to the loglikelihood (\ref{EQ:log.lik}), that is,
\begin{equation}\label{EQ:shrinkage.covglasso.estimator}
\widehat{\Sigma}=
\argmax_{\Sigma\in\mathcal{S}_{\succ}(\G)}
\left\{
\ell_{Y}(\Sigma) - \Ps_{\lambda,\kappa}(\Sigma)
\right\}.
\end{equation}
We call this the \texttt{ridge-regularized covglasso estimator} because it comprises both a lasso and a ridge regularization.  Indeed, as we will clarify below, when $\lambda=0$ it results in a ridge-regularized estimate of $\Sigma$, and for this reason we will refer to $\kappa$ as the ridge penalty term. On the other hand, it is clear that when $\kappa=0$ it simplifies to the covglasso approach without diagonal penalization. It is also worth remarking that in (\ref{EQ:shrinkage.covglasso.estimator}) we consider the optimization of the objective function over $\mathcal{S}_{\succ}(\G)$, so that the ridge-regularized covglasso estimator will be a matrix adapted to the graph $\G$. In this way, we deal explicitly also with the case where some specific entries of $\Sigma$ are constrained to be equal to zero. When $\G$ is set to the complete graph, then no zero constraints are set a priori.
Putting everything together, the ridge-regularized covglasso estimator in (\ref{EQ:shrinkage.covglasso.estimator}) involves three parameters: $\lambda$, $\kappa$ and $\G$. By considering the possible combinations resulting from setting  $\lambda$ and/or $\kappa$ to zero, and/or $\G$ to the complete graph results in $2^{3}=8$ specific methods, which are detailed in  Table~\ref{TAB:different.types.estimators}. Some of these cases are already available in the literature and have been unified within the same framework, while some combinations are new. The result is a flexible tool that allows to apply, individually and jointly, a number of different approaches, and to perform model selection within the full space of models.
{
\begin{table}
\caption{Approaches to the estimation of $\Sigma$ resulting from the application of the ridge-regularized covglasso method with different settings of the parameters $\lambda$, $\kappa$ and $\G$. In each row, the symbols ``$\times$'' identify which constraints are applied.}
\label{TAB:different.types.estimators}
\begin{center}
\begin{tabular}{|N|*{3}{c|}p{23eM}|}
\hline
\multicolumn{1}{|c}{\multirow{2}{*}{}} & \multicolumn{3}{|c|}{Parameter settings} &  \multicolumn{1}{c|}{\multirow{2}{*}{Type of estimator}}\\
\cline{2-4}\setcounter{rowcntr}{0}
& $\lambda=0$ & $\kappa=0$ & $\G$ complete & \\
\hline\hline
\label{TABROW:different.types.estimators.mle} 1\;) &  $\times$   & $\times$ & $\times$ &  MLE for $\Sigma\in\mathcal{S}_{\succ}$\\
\hline
 \label{TABROW:different.types.estimators.hybmle} 2\;) &  $\times$   & $\times$ &  &  MLE for $\Sigma\in\mathcal{S}_{\succ}(\G)$ \citep{chaudhuri2007estimation}\\
\hline
\label{TABROW:different.types.estimators.rrmle} 3\;) & $\times$   &  & $\times$ &  ridge-regularized MLE for $\Sigma\in\mathcal{S}_{\succ}$ \citep{warton2008penalized}\\
\hline
\label{TABROW:different.types.estimators.cvg} 4\;) &  & $\times$ & $\times$ & covglasso for $\Sigma\in\mathcal{S}_{\succ}$ \citep{lam2009sparsistency} \\
\hline
\label{TABROW:different.types.estimators.hybcvg} 5\;) & & $\times$ &  &  covglasso for $\Sigma\in\mathcal{S}_{\succ}(\G)$  \citep{bien2011sparse}\\
\hline
\label{TABROW:different.types.estimators.hybrrmle} 6\;)  & $\times$ & & &  ridge-regularized MLE for $\Sigma\in\mathcal{S}_{\succ}(\G)$ \\
\hline
\label{TABROW:different.types.estimators.rrcvg} 7\;)  & & & $\times$ & ridge-regularized covglasso for $\Sigma\in\mathcal{S}_{\succ}$\\
\hline
\label{TABROW:different.types.estimators.hybrrcvg} 8\;) &  &  &  & ridge-regularized covglasso for $\Sigma\in\mathcal{S}_{\succ}(\G)$\\
\hline
\end{tabular}
\end{center}
\end{table}
}

In the rest of this section, we first discuss the motivation for replacing $\diagOP(\Sigma)$ with $\diagOP(\Sigma^{-1})$ in the penalty term. We then make the connection with the ridge-regularization introduced by $\kappa>0$. Finally, we provide a formulation of the objective function in (\ref{EQ:shrinkage.covglasso.estimator}) to be used in the algorithm for the computation of the estimator.

In order to motivate the $\ell_{1}$ penalization of the diagonal of $\Sigma^{-1}$, a key role is played by the relationship between $\tau_{i}$ in (\ref{EQ:tau}) and the $i$-th diagonal element of the inverse covariance matrix, $\theta_{ii}$, in (\ref{EQ:shrinkage-covglasso.penalty}). In particular, it holds that $\theta_{ii}=1/\tau_{i}$, for every $i\in V$ \citep[see][Corollary~5.8.1]{whittaker1990graphical}. With this, we can rewrite the penalty term  (\ref{EQ:shrinkage-covglasso.penalty}) in the form
\begin{align}\label{EQ:shrinkage-covglasso.penalty.2}
\Ps_{\lambda,\kappa}(\Sigma)= \lambda\,\sum_{j,k\in V; j\neq k} |\sigma_{jk}| + \kappa\,\sum_{j\in V} \frac{1}{\tau_{j}},
\end{align}
that is the proposed approach penalizes small $\tau_i$ values.  The parameter $\tau_i$ plays the role of the error variance in the regression component.  In order to see this more clearly, we extract
from the penalty $\Ps_{\lambda,\kappa}(\Sigma)$ in (\ref{EQ:shrinkage-covglasso.penalty.2})
only the elements $\Ps_{\lambda,\kappa}(\tau_{i}, \sigmaii)$  that involve the parameters $\sigmaii$ and $\tau_{i}$. Then, similarly to the loglikelihood decomposition in (\ref{EQ:log.lik.split}),
we write the penalized loglikelihood regression component of the objective function in  (\ref{EQ:shrinkage.covglasso.estimator}) as
\begin{align}\label{EQ:pen.cond.log.lik}
\ell_{\vctr{y}_{i}\mid Z_{-i}}(\tau_{i}, \sigmaii) - \Ps_{\lambda,\kappa}(\tau_{i}, \sigmaii)
= -\log{\tau_i} - \frac{1}{n\tau_i}\|\vctr{y}_i -  Z_{-i}\sigmaii\|^2_2
 - 2\lambda\,\|\sigmaii\|_{1}  - \kappa\times \frac{1}{\tau_{i}}.
\end{align}
When $\tau_{i}$ is known, the  optimization of (\ref{EQ:pen.cond.log.lik}) with respect to the regression coefficients $\sigmaii$
is the classical lasso regression problem \citep{tibshirani1996regression}. When $\tau_i$ is not known, specific methods designed to jointly estimate the regression coefficients and the error variance are considered. Several papers in the area of lasso regression discuss the difficulty of estimating the error variance, as it tends to be underestimated \citep{fan2012variance,reid2016study,yu2019estimating}. Hence, penalizing small $\tau_i$ values has the effect of counteracting such downward bias. In (\ref{EQ:pen.cond.log.lik}) this is done by the inclusion of the additional penalty term $1/\tau_{i}$, which results from the $\ell_{1}$-penalty on the diagonal of $\Sigma^{-1}$. On the other hand, the $\ell_{1}$-penalty on the diagonal of $\Sigma$, as in the covglasso, has the opposite effect of penalizing large $\tau_i$ values, as it can be seen in \citet[][equation~(4)]{wang2014coordinate}.
Interestingly, (\ref{EQ:pen.cond.log.lik}) belongs to the family of penalized loglikelihood functions introduced in \citet{yu2019estimating}, where the penalty term is a convex function of the natural parameters of the Gaussian multi-parameter exponential family.  Within their general family, \citet{yu2019estimating} investigate the properties of the estimates resulting from the penalty terms $\lambda\,\|\sigmaii\|_{1}/\tau_{i}$ and $\lambda\,\|\sigmaii\|_{1}^{2}/\tau_{i}$, and they both involve penalizing small error variances. A similar idea was previously considered by \citet{stadler2010ell}, who proposed an $\ell_{1}$-norm penalty on $\sigmaii/\sqrt{\tau_{i}}$ and noted that their approach has a close relation with the Bayesian lasso \citep{park2008bayesian}.

An additional advantage of the proposed penalty is that the penalization on the diagonal of the inverse covariance results in a loglikelihood where the empirical covariance matrix is replaced by its ridge-regularized version $(S+\kappa\, I_{p})$, which is positive definite for every $\kappa>0$. As a consequence,  the optimization problem is well-defined also when $S$ has not full rank. This is formally stated in equation (\ref{EQ:shrinkage.2}) of the following proposition.
\begin{prop}\label{THM:shrinkage.transform}
If we set $Y^{(\kappa)}=\left(Y^\top, \sqrt{n\kappa}\, I_{p}\right)^{\top}$ then the objective function in
(\ref{EQ:shrinkage.covglasso.estimator}) can be written as
\begin{align}
\ell_{Y}(\Sigma) - \Ps_{\lambda,\kappa}(\Sigma)
\label{EQ:shrinkage.1}
&= \ell_{Y^{(\kappa)},n}(\Sigma) - \lambda\,\|\Sigma-\diag{\Sigma}\|_{1}\\[1ex]
\label{EQ:shrinkage.2}
&= -\log{|\Sigma|} -\tr(\Sigma^{-1}S^{(\kappa)}) - \lambda\,\|\Sigma-\diag{\Sigma}\|_{1}
\end{align}
where $S^{(\kappa)}=\left(S+\kappa\, I_{p}\right)$ is a ridge-regularized version of $S$.
\end{prop}
\begin{proof}
See Appendix~\ref{PROOF:THM:shrinkage.transform}.
\end{proof}

In the case where $\lambda=0$, equation  (\ref{EQ:shrinkage.2}) in Proposition~\ref{THM:shrinkage.transform} coincides with the result of \citet[][Corollary~1]{warton2008penalized} where it is shown that  $S^{(\kappa)}$ is the maximum penalized loglikelihood estimator of $\Sigma$ when the penalty term is $\kappa\tr(\Sigma^{-1})$. Hence, in a similar way to  \citet{warton2008penalized}, we will exploit this result to approach in a formal way the problem of determining the optimal amount of ridge-regularization $\kappa$.

Finally, we turn to the practical problem of the computation of the ridge-regularized covglasso estimator in (\ref{EQ:shrinkage.covglasso.estimator}). As shown by
\citet{wang2014coordinate}, the covglasso estimator can be computed by iteratively maximizing the penalized conditional loglikelihood, in a similar way to the ICF algorithm for the computation of MLEs. However, differently to the diagonal entries of $\Sigma$ which are computed on the univariate marginal distributions, each error variance $\tau_{i}$, for $i=1,\ldots, p$, depends on the joint distributions of all the variables. For this reason, the solution to (\ref{EQ:shrinkage.covglasso.estimator}) cannot be computed by iteratively maximization of the building blocks (\ref{EQ:pen.cond.log.lik}). Proposition~\ref{THM:shrinkage.transform} shows how  problem (\ref{EQ:shrinkage.covglasso.estimator}) can be reformulated into an $\ell_{1}$-penalized problem on an augmented data matrix. More specifically,
equation  (\ref{EQ:shrinkage.1}) shows that the ridge-regularized covglasso estimator coincides with the covglasso estimator computed from (\ref{EQ:covglasso.penalty}) when $Y$ is replaced by the extended data matrix $Y^{(\kappa)}$, but $n$ is left unchanged.  It is worth recalling that an extended matrix similar to $Y^{(\kappa)}$ is used to relate ridge regression with the ordinary least-squares method; see, among others, \citep[][Section~4]{hastie2020ridge}. We can thus provide a decomposition of the objective function in (\ref{EQ:shrinkage.covglasso.estimator}), into a marginal and a conditional penalized loglikelihood,  that both exploits (\ref{EQ:shrinkage.1}) and explicitly deals with the constraint $\Sigma\in\mathcal{S}_{\succ}(\G)$.
\begin{corollary}\label{THM:final.loglik.decomposition} Consider the undirected graph  $\G=(V,E)$ and assume that $\Sigma\in\mathcal{S}_{\succ}(\G)$. Then for every $i\in V$ the  objective function in
(\ref{EQ:shrinkage.covglasso.estimator}) can be decomposed, similarly to (\ref{EQ:log.lik.split}), into a marginal and conditional penalized loglikelihood as follows,
\begin{align}
\ell_{Y}(\Sigma) - \Ps_{\lambda, \kappa}(\Sigma)
&= \ell_{Y^{(\kappa)}_{-i}, n}(\Sigmaii) - \lambda\|\Sigmaii - \diag{\Sigmaii}\|_1 +\nonumber\\
&\phantom{=} -\log{\tau_i} - \frac{1}{n\tau_i}\|\vctr{y}^{(\kappa)}_i - Z^{(\kappa)}_{\bnd{i}}\sigmabdi\|_2^2 - 2\lambda\|\sigmabdi\|_1, \label{EQ:lin.reg.touse}
\end{align}
where $Y^{(\kappa)}$ is as in Proposition~\ref{THM:shrinkage.transform},
$Z^{(\kappa)}_{\bnd{i}}$ is the submatrix of $Z^{(\kappa)}_{-i} = Y_{-i}^{(\kappa)}\Sigmaiinv$ made up of the columns  indexed by the elements of $\bnd{i}$, and similarly for the subvector $\sigmabdi$ of $\sigmaii$.
\end{corollary}
\begin{proof}
See Appendix~\ref{PROOF:THM:final.loglik.decomposition}.
\end{proof}

As we will see in the next section, Corollary~\ref{THM:final.loglik.decomposition} implies that the computation of the ridge-regularized covglasso estimator can be reconduced to the existing efficient computational procedure of \citet{wang2014coordinate}.

\section{\algonamefull{} (\algoname) algorithm}\label{SEC:hybrid.procedure}
In this section we present the algorithm for the computation of (\ref{EQ:shrinkage.covglasso.estimator}). We name it Generalized Iterative Conditional Fitting (GICF), because it is an extension
of the ICF algorithm of \citet{chaudhuri2007estimation}.
In particular, the GICF algorithm generalizes the ICF by replacing the regression component of (\ref{EQ:log.lik.split}) with (\ref{EQ:lin.reg.touse}). This allows the inclusion of a lasso penalty on the regression coefficient \citep{wang2014coordinate}, the enforcement of  a constraint $\Sigma\in \Set_\succ(\G)$ and the penalization of the error variance. The latter is achieved implicitly by replacing $Y$ with $Y^{(\kappa)}$, as outlined in Proposition~\ref{THM:shrinkage.transform}.

Going into the details of the algorithm, we first consider the numerical maximization of the linear regression component (\ref{EQ:lin.reg.touse}). Although this function is not jointly concave in
$\tau_i$ and $\sigmabdi$ \citep[see][Section 3.1]{stadler2010ell}, as noted by \citet{wang2014coordinate}, it is still possible to employ a coordinate-descent algorithm to efficiently solve the optimization problem. This means that we can proceed by optimizing the objective function with respect to only one coordinate at a time, while keeping the others fixed.
More concretely, we first notice that, for a fixed $\sigmabdi$, the optimum with respect to $\tau_i$ can be straightforwardly computed from the partial derivative, leading to
\begin{equation}\label{EQ:cd.tau}
\hat{\tau}_i = \frac{1}{n}\|\vctr{y}^{(\kappa)}_i - Z^{(\kappa)}_{\bnd{i}}\sigmabdi\|_2^2.
\end{equation}
As for the optimum along the direction of any entry $\sigma_{ij}$ of $\sigmabdi$, we remark that, because of the $\ell_1$-penalty, the objective function is non-differentiable. However, it still concave with respect to $\sigma_{ij}$. Therefore, in line with \citet{wang2014coordinate}, we apply the theory of subdifferentials \citep{rockafellar1970convex}, which yields
\begin{equation}\label{EQ:cd.sigma}
\hat{\sigma}_{ij} = \frac{1}{C_{ij}}\mathcal{S}_{\lambda}\left(B_{ij}\right),
\end{equation}
where
\begin{gather*}
B_{ij} = \frac{1}{n\tau_i}\bigg(\left[(Z^{(\kappa)}_{\bnd{i}})^\top \vctr{y}^{(\kappa)}_i\right]_j - \sum_{\substack{m \in \bnd{i}\\ m \neq j}}\left[(Z^{(\kappa)}_{\bnd{i}})^\top Z^{(\kappa)}_{\bnd{i}}\right]_{jm}\sigma_{im}\bigg), \\[1ex]
 C_{ij} =\frac{1}{n\tau_i}\left[(Z^{(\kappa)}_{\bnd{i}})^\top Z^{(\kappa)}_{\bnd{i}}\right]_{jj}
\end{gather*}
and $\mathcal{S}_\lambda(x) = \text{sign}(x)(|x| - \lambda)_+$ is the soft-thresholding operator; it is worth recalling that the matrix $Z^{(\kappa)}_{\bnd{i}}$ has full rank for any $\kappa\geq 0$ when $n>p$ and for $\kappa > 0$ when $n \leq p$ and, thus, the quantity $\hat{\sigma}_{ij}$ is therefore well-defined under the aforementioned conditions.

The GICF algorithm optimizes (\ref{EQ:shrinkage.covglasso.estimator}) by iteratively decomposing the objective function as in Corollary~\ref{THM:final.loglik.decomposition}. Then, at every step, it maximizes the penalized regression component as described above. This is shortly outlined below, whereas the detailed pseudocode, divided between the main cycle and the intermediate linear regression phase, can be found in  Algorithms~\ref{ALG.algo.detailed} and \ref{ALG.cd.lin.detailed} of Appendix~\ref{APP:pseudocodes}:
\begin{enumerate}
\item Set the initial estimate of $\Sigma$ to $\widehat{\Sigma} = \widehat{\Sigma}^{(0)}$; for example $\widehat{\Sigma}^{(0)} = \diagOP(S^{(\kappa)})$, and iterate the following point until convergence.
\item For each $i \in V$:
	\begin{enumerate}
	\item fix the submatrix $\widehat{\Sigma}_{-i,-i}$ and compute the pseudo-predictors $\widehat{Z}^{(\kappa)}_{\bnd{i}}$ as in Corollary~\ref{THM:final.loglik.decomposition};
	\item compute $\hat{\tau}_i$ and $\sigmabdiest$ by iteratively repeating (\ref{EQ:cd.tau}) and (\ref{EQ:cd.sigma}) until convergence;
	\item update $\widehat{\Sigma}$ by setting $\hat{\vctr{\sigma}}_{-i,i} = (\sigmabdiest^\top,\, \vctr{0}^\top)^\top$ and $\hat{\sigma}_{ii} = \hat{\tau}_i + \hat{\vctr{\sigma}}_{-i,i}^\top\left(\widehat{\Sigma}_{-i,-i}\right)^{-1}\hat{\vctr{\sigma}}_{-i,i}$.
	\end{enumerate}
\end{enumerate}

We turn now to the convergence of the iterative procedure. Firstly, we note that Proposition~\ref{THM:shrinkage.transform} enables us to treat every instance of the optimization problem
(\ref{EQ:shrinkage.covglasso.estimator}) as in the particular case where $\kappa = 0$. As a consequence, when $\lambda = 0$, convergence is granted by the same result employed by \citet{chaudhuri2007estimation}, which, in turn, follows from \citet[Proposition 1]{drton2006maximum}. In the general case, where $\lambda>0$, we note that the arguments used in \citet{wang2014coordinate}, based on \citet[Theorem 4.1]{tseng2001convergence}, remain valid also under the  constraint $\Sigma \in \Set_\succ(\G)$.
This is formalized by the following proposition.

\begin{prop}\label{PROP.convergence}
Let $\mathcal{E} = \{\widehat{\Sigma}^{(r)}\}_{r\in\N}$ be the sequence of estimates produced by \algoname{} and $\mathcal{V} =\big\{ \ell_{Y}(\widehat{\Sigma}^{(r)}) - \Ps_{\lambda,\kappa}(\widehat{\Sigma}^{(r)})\big\}_{r \in \N}$ be the sequence of values of the objective function of (\ref{EQ:shrinkage.covglasso.estimator}). Assume that either $n > p$ or $\kappa > 0$, so that $S^{(\kappa)}$ is positive definite. Then, each accumulation point of $\mathcal{E}$ is a stationary point of $\ell_{Y}(\Sigma) - \Ps_{\lambda,\kappa}(\Sigma)$ and $\mathcal{V}$ converges to a finite value $\mathcal{V}_\infty$.
\end{prop}

Notice how in Proposition~\ref{PROP.convergence}, for $\lambda=0$, the result is true with respect to the classical notion of stationarity, whereas when $\lambda>0$ the objective function is non-differentiable and the result refers to directional stationarity, or d-stationarity for short; see (\citeauthor{li2020stationarity}, \citeyear{li2020stationarity}; \citeauthor{cui2021optimization}, \citeyear{cui2021optimization}, Chapter 6).

\section{Selection of penalty parameters and application of the method}\label{SCT:model.selection}
This section deals with the practical application of the method. More specifically, we provide the maximum theoretical values of the penalty parameters and give details on the cross-validation approach for the selection of the model.

The application of penalized likelihood methods requires the initial definition of a grid of penalty parameter values to obtain a path of solutions. To this aim, it is useful to identify the values of the pair $(\lambda, \kappa)$ that return a diagonal solution
so as to identify the suitable penalty parameter region within the domain $[0; +\infty)^{2}$. As it turns out, it is possible to exploit the structure of the \algoname\ algorithm to study the conditions under which the estimate becomes diagonal, as well as the interplay between the parameters $\kappa$ and $\lambda$. We first consider $\kappa$ fixed and provide the maximum theoretical value of the covglasso penalty.
\begin{theorem}\label{THM.diag.S}
Consider the ridge-regularized covglasso problem in (\ref{EQ:shrinkage.covglasso.estimator}).
For a given graph $\G = (V, E)$ and $\kappa \geq 0$ let,
\begin{align*}
  \lambda_{MAX}(\kappa) = \max_{\{i,j\}\in\G}\frac{|s_{ij}|}{(s_{ii} + \kappa)(s_{jj} + \kappa)}
\end{align*}
for  $E \neq \emptyset$, and $\lambda_{MAX}(\kappa)=0$ otherwise. Then, for every $\lambda \geq \lambda_{MAX}(\kappa)$ the following statements hold true,
\begin{enumerate}
\item[(i)] $\diag{S^{(\kappa)}}$ is a stationary point of the objective function of (\ref{EQ:shrinkage.covglasso.estimator}) in $\mathcal{S}_{\succ}(\G)$;
\item[(ii)] if $\widehat{\Sigma}^{(0)} = \diag{S^{(\kappa)}}$, then the output of the \algoname\ algorithm is $\widehat{\Sigma} = \diag{S^{(\kappa)}}$.
\end{enumerate}
\end{theorem}
\begin{proof}
See Appendix~\ref{AX:proof.diag.S}
\end{proof}

Theorem~\ref{THM.diag.S} shows that for $\lambda$ sufficiently large the \algoname\ algorithm returns a diagonal solution. A different role is played by the ridge penalty $\kappa$. Indeed, if the covglasso penalty is not included, that is, if $\lambda$ is set to zero, then the ridge-regularized estimator is never diagonal, unless $S$ is diagonal or a diagonal solution is forced by setting $\{i,j\}\not\in \G$ for every $s_{ij}\neq 0$.  However, as shown in Theorem~\ref{THM:kappa_max} below,  $\lambda_{MAX}(\kappa)$  is a decreasing function of $\kappa$ so that when $\lambda$ is greater than zero, it is always possible to set $\kappa$ to a value sufficiently large to imply that $\lambda_{MAX}(\kappa)$ is smaller than $\lambda$ and, consequently, the solution is diagonal by  Theorem~\ref{THM.diag.S}. Hence, the values of $\lambda$ and $\kappa$ cannot be chosen independently of each other, and we now formally state the relationship existing between these two quantities.
\begin{theorem}\label{THM:kappa_max}
In the framework of Theorem~\ref{THM.diag.S}, assume that $\lambda_{MAX}(\kappa)$ is not identically zero. Then $\lambda_{MAX}(\kappa)$ is strictly decreasing and approaches $0$ as $\kappa$ increases, so that $\lambda_{MAX}(\kappa) \leq \lambda_{MAX}(0)$ for every $\kappa\geq 0$. Furthermore, for every $0 < \lambda \leq \lambda_{MAX}(0)$, it holds that
\begin{align*}
\lambda \leq \lambda_{MAX}(\kappa)
\quad \Longleftrightarrow\quad
\kappa \leq \kappa_{MAX}(\lambda),
\end{align*}
where
\begin{align*}
\kappa_{MAX}(\lambda) = \max_{\substack{\{i,j\} \in \G\\g_{ij}(\lambda) \geq 0}}\left\{\sqrt{\frac{1}{4}(s_{ii} + s_{jj}) + g_{ij}(\lambda)} - \frac{1}{2}(s_{ii} + s_{jj})\right\}
\end{align*}
with $g_{ij}(\lambda) = \frac{|s_{ij}|}{\lambda} - s_{ii}s_{jj}$. Finally, the set $\big\{\{i,j\} \in \G: g_{ij}(\lambda) \geq 0\big\}$ is not empty and, consequently, $\kappa_{MAX}(\lambda)$ is well-defined.
\end{theorem}
\begin{proof}
See Appendix~\ref{AX:proof.kappa.max}
\end{proof}
Theorem~\ref{THM:kappa_max} can be readily applied to identify a set of suitable penalty values. Indeed, it shows that for every $\lambda\geq \lambda_{MAX}(0)$ the solution is diagonal, regardless of the value of $\kappa$, thereby providing
an upper bound to the suitable $\lambda$ values. Furthermore, for every $0 < \lambda\leq \lambda_{MAX}(0)$ the solution becomes diagonal for every $\kappa\geq \kappa_{MAX}(\lambda)$ so that it makes sense to search for  the optimal values of $\lambda$ and $\kappa$  within the set,
\begin{align}\label{EQN:feasible.set}
\Hs =  \left\{(\lambda, \kappa) \in \R^2: 0 \leq \lambda \leq \lambda_{MAX}(0) \mbox{ and } 0\leq \kappa \leq \kappa_{MAX}(\lambda)\right\},
\end{align}
where for $\lambda=0$ the inequality $0\leq \kappa \leq \kappa_{MAX}(0)$ should be read as $0\leq \kappa<\infty$.
Equation (\ref{EQN:feasible.set}) defines the set $\Hs$ in a way that facilitates the construction of a finite subset $\Hs^{\prime}\subset \Hs$ of parameter values because it provides upper bounds to both the $\lambda$ and the $\kappa$ values.  In this work, we obtain $\Hs^{\prime}$ as follows. Firstly, we specify an increasing sequence $\lambda_{1},\ldots,\lambda_{r}$
of $r$ values equally spaced between $0$ and $\lambda_{MAX}(0)$. Then every $\lambda_{i}$, for $i=1,\ldots, r$, is paired with the elements of an increasing sequence $\kappa_{1},\ldots,\kappa_{s_{i}}$ of $s_{i}$ values equally spaced between $0$ and $\kappa_{MAX}(\lambda_{i})$. We initially assign to $r$ and $s_{1}$ two (large) integer values and then, for every $i=2,\ldots, r$, we set $s_{i}$ to the smallest integer larger than $s_{1}\times \kappa_{MAX}(\lambda_{i})/\kappa_{MAX}(\lambda_{1})$. In this way, for every $\lambda_{i}$ the number of $\kappa$ values is reduced proportionally to the size of the interval of suitable values. In the case where $\lambda_{1}=0$, in the procedure for the specification of the $\kappa$ values, we replace $\lambda_{1}$ with a positive value smaller than $\lambda_{2}$.

As model selection criterion, we employ an $M$-fold cross validation. To this end, let $\mathcal{C}_1, \dots, \mathcal{C}_{M}$ be a partition of the set of indices $\mathcal{N} = \{1, \dots, n\}$ and let $Y(\mathcal{C}_m) \in \R^{|\mathcal{C}_m| \times p}$ be the submatrix of $Y$ obtained by selecting the rows indexed by the elements of $\mathcal{C}_m$. Also, let $\widehat{\Sigma}_{\kappa, \lambda}(\mathcal{C}_m)$ be the estimate of $\Sigma$ produced by the \algoname\ algorithm, as in Algorithm~\ref{ALG.algo.detailed}, when provided with the data matrix $Y(\mathcal{C}_m)$ and the penalty coefficients $\kappa$ and $\lambda$. Hence, following, among others, \citet{warton2008penalized} and \citet{bien2011sparse}, we consider the cross-validated loglikelihood,
\begin{equation}\label{EQ:cv.criterion}
(\hat{\kappa}, \hat{\lambda}) = \argmax_{(\kappa, \lambda) \in \Hs^{\prime}} \sum_{m = 1}^M\ell_{Y(\mathcal{C}_m)}\{\widehat{\Sigma}_{\kappa, \lambda}(\mathcal{N} \setminus \mathcal{C}_m)\}.
\end{equation}

Figure~\ref{IMG:cv.profile} provides an example of the profile of the function over $\Hs^{\prime}$, based on a set of simulated data. The dashed line represents $\lambda_{MAX}(\kappa)$. This, as expected, is a decreasing function of $\kappa$. Moreover, it appears that the maximum theoretical value of $\lambda$ decreases considerably as $\kappa$ increases. This suggests that the identification of the set $\Hs$ may effectively reduce the dimension of the search space and thus improve efficiency. It is also interesting to note that if $\kappa$ is na\"ively set to a value close to zero, then a local optimum is found at a large value of $\lambda$, whereas if $\kappa$ is allowed to vary over a larger interval, then a better solution is identified at a lower value of $\lambda$. We will see in the simulations of Section~\ref{SCT:simulations} that this will typically correspond to a less sparse solution.
\begin{figure}
\centering
\includegraphics[width = 0.6\linewidth]{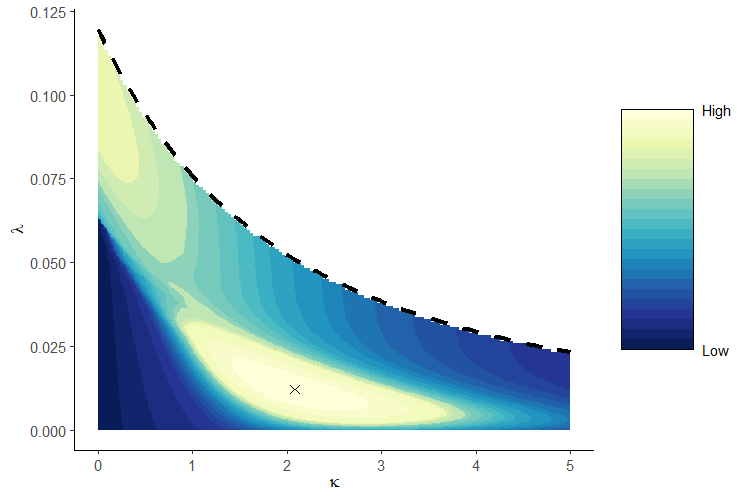}
\caption{An example of the profile of the cross-validated loglikelihood (\ref{EQ:cv.criterion}), based on simulated data. The dashed line is the function $\lambda_{MAX}(\kappa)$: over and above this line the solution is diagonal. The `$\times$' symbol denotes the global optimum.}\label{IMG:cv.profile}
\end{figure}
\section{Simulation study}\label{SCT:simulations}

In this section, we evaluate the performance of the proposed method via a simulation study. We highlight the specific cases mentioned  in Table~\ref{TAB:different.types.estimators} and focus on the added benefits of having a unified approach. In particular, our proposed method includes a large class of models that can be explored via the parameters $\lambda$, $\kappa$ and $\G$, as opposed to more standard techniques which generally focus only on specific cases within this class.

\subsection{Simulation setting}\label{SCT:simulation.setting}
Across all settings, we simulate data from a multivariate normal distribution with mean $\vctr{0}$ and a covariance matrix $\Sigma \in \R^{p \times p}$. The covariance matrix is taken to have a banded structure, that is, $\sigma_{ij} = 0$ whenever $|i - j| > b$, with $b$ an integer value which will be chosen so that the graph $\G$ which encodes the true sparsity pattern of $\Sigma$ has a predetermined density. The non-zero off-diagonal entries are randomly selected in \{-1, 1\} with equal probability, whereas the diagonal values are fixed to that constant value which makes the condition number of $\Sigma$,  i.e., the ratio between its largest and smallest eigenvalue, equal to $p$. Such a method for the construction of the diagonal has been employed, among others, also by \citet{bien2011sparse}.

\subsection{Ridge-regularized estimator under a graph \boldmath $\G$}\label{SCT:simul.mle}
In this simulation study, we consider the case when a graph $\G$ encodes a given sparsity pattern of the covariance matrix $\Sigma$. The aim is to estimate $\Sigma$ within the space $\Set_{\succ}(\G)$. We are therefore particularly interested in the comparison between the  MLE of $\Sigma$ in $\Set_{\succ}(\G)$, which can be computed using the ICF algorithm of \citet{chaudhuri2007estimation} when $n>p$ (case \ref{TABROW:different.types.estimators.hybmle} in Table~\ref{TAB:different.types.estimators}), and  our ridge-regularized estimator of $\Sigma$ (case \ref{TABROW:different.types.estimators.hybrrmle} in Table~\ref{TAB:different.types.estimators}), which is an extension of \citet{warton2008penalized} (case \ref{TABROW:different.types.estimators.rrmle} in Table \ref{TAB:different.types.estimators}) to the case of a pre-specified sparsity pattern $\G$ and is available also in the high dimensional case.

For the simulation, we fix the number of variables to $p=50$ and consider a banded covariance matrix $\Sigma$, with a density varying across three levels ($10\%, 30\%$ and $50\%$). We then simulate data with sample sizes varying in $n \in \{45, 75, 100, 250, 500, 1000\}$, so as to consider both the high dimensional ($p>n$) and low dimensional ($p<n$) cases.  The MLEs of $\Sigma$ have been computed using the ICF algorithm only when $n > p$, while the ridge-regularized estimators have been computed for all cases. For the latter, we set $\lambda=0$, since $\G$ is given, and select $\kappa$ via $5$-fold cross-validation from a uniform grid of $30$ values in the interval $[0, 3]$.

We repeat the simulations 10 times and report the average values in  Figure~\ref{IMG.sim.mle}.
\begin{figure}[!htb]
\centering
	\hspace{-1.2cm}
	\begin{subfigure}[t]{0.47\textwidth}
	\centering
	\caption{\label{IMG.sim.mlea}}
	\includegraphics[width = \linewidth]{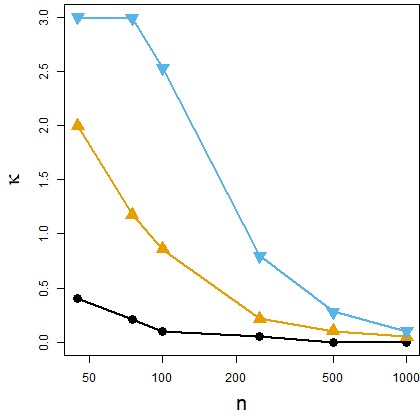}
	\end{subfigure}
	\hspace{0.5cm}
	\begin{subfigure}[t]{0.47\textwidth}
	\centering
	\caption{\label{IMG.sim.mleb}}
	\includegraphics[width = \linewidth]{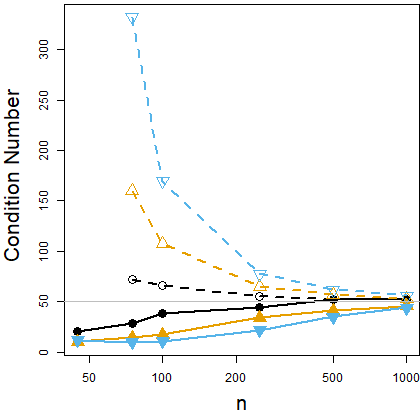}
	\end{subfigure}

	\includegraphics[width=0.4\textwidth]{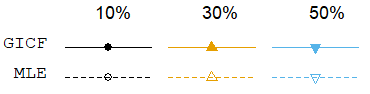}
	
\caption{Average results from 10 simulations with sample size $n \in \{45, 75, 100, 250, 500, 1000\}$, $p=50$ and banded covariance matrices $\Sigma$ associated to graphs $\G$ with varying density levels (10\%, 30\%, 50\%). Comparison between the ridge-regularized \algoname{} ($\kappa \geq 0,\, \lambda = 0$) and the MLE ($\kappa = 0,\, \lambda = 0$) in $\Set_\succ(\G)$ shows how  (a) the optimal $\kappa$ decreases with $n$, i.e., ridge regularization is particularly useful for small $n$ and (b) the ridge-regularized $\widehat{\Sigma}$ has a condition number closer to the true one (horizontal line) than the MLE, i.e., it is a more stable estimate.}\label{IMG.sim.mle}
\end{figure}
Figure~\ref{IMG.sim.mlea} shows how the optimal value of $\kappa$ decreases as the sample size $n$ increases. This is to be expected as stabilization of the sample covariance becomes less important the larger $n$ is. The fact that $\kappa$ goes to $0$ as $n$ goes to infinity is formally shown by \citet[Theorem 2]{warton2008penalized} for the case of a complete graph $\G$. More importantly, the results show how the regularization induced by $\kappa$, which has the effect of rendering the matrix $S^{(\kappa)}$ nonsingular even when $S$ is, and is therefore essential when $n < p$, retains its importance also when $n$ is greater than $p$, and can be neglected only when $n$ becomes very large compared to $p$.
Figure~\ref{IMG.sim.mleb} reports the condition number of $\widehat{\Sigma}$, giving an indication of the stability of the estimate \citep{won13}. The results show how the MLE estimators, which are calculated when $n>p$, largely overestimate the true condition number when $n$ is small. In contrast to this, the ridge-regularized estimators are significantly more stable than the MLE estimators and are well-conditioned also in the high dimensional case of $n<p$, despite some underestimation. As $n$ increases, both estimators converge to the true condition number (horizontal line). As regards to the three levels of graph sparsity, we see qualitatively similar results in terms of the behaviour of the optimal $\kappa$ against $n$ for the different levels. In addition, the results show how the denser the graph is, the less stable the estimator is, and thus the more important ridge-regularization becomes. Putting all results together, this simulation study shows how the ridge regularization induced by the tuning parameter $\kappa$ is important also when $n$ is larger than $p$, and how the optimal level of regularization depends both on the sample size $n$ and on the graph sparsity.

\subsection{Covglasso estimator under a partially known graph \boldmath $\G$}\label{SCT:simul.time}
We now focus on the case of $\kappa=0$ where our approach reduces to the covglasso of \citet{lam2009sparsistency}. 
We focus in particular on the case where a sparsity pattern $\G$ is given. While this case is naturally embedded within our framework, the implementation of covglasso using the algorithm of \citet{wang2014coordinate}, which is available in the \texttt{R} package \texttt{covglasso}, can also be tuned to this case. In particular, as suggested by \citet[Section 6]{bien2011sparse}, one can set sufficiently large penalties to those entries of the covariance matrix which are  $0$ in $\G$. As we will argue in this section, our proposed algorithm is computationally more efficient than the \texttt{covglasso} algorithm in this case.

As before, we consider banded covariance matrices, with varying $p \in \{25, 50, 100, 200, 300\}$ and a number of bands $b \in \{3, 10\}$. We choose a fixed number of bands in order to maintain a constant neighbourhood size. For each setting, 10 datasets of $n = 2000$ observations are simulated.
For simplicity, we compare MLE estimates in $\Set_\succ(\G)$ for both methods. To this end, we set $\lambda=0$ in our GICF algorithm (row \ref{TABROW:different.types.estimators.hybmle} in Table~\ref{TAB:different.types.estimators}), i.e., our implementation of the ICF algorithm of \citet{chaudhuri2007estimation}, while for \texttt{covglasso}, we set $\lambda=0$ for all non-zero off-diagonal elements of $\Sigma$ and $\lambda = 10^6$ for those entries which are $0$ in $\G$.

Table~\ref{TBL:simul.time} reports the average computational time required to reach convergence across the different settings and algorithms.
\begin{table}[!tb]
\centering
\caption{Average computational times (in seconds) for \algoname{} and \texttt{covglasso} from 10 simulations with sample size $n=2000$, $p \in \{25, 50, 100, 200, 300\}$ and covariance $\Sigma$ associated to graphs $\G$  with $b \in \{3, 10\}$ number of bands. Comparison between the \algoname{} ($\kappa = 0,\, \lambda = 0$) and the \texttt{covglasso} algorithms for the calculation of MLE in $\Set_\succ(\G)$ shows how \algoname{} is computationally more efficient.}\label{TBL:simul.time}
\begin{tabular}{|c| c |c|c| c |c|c|}
\hline
$p$  & & \multicolumn{2}{|c|}{3 Bands} & & \multicolumn{2}{|c|}{10 Bands}\\
\cline{3 - 7}
 &  & {\footnotesize \algoname} & {\footnotesize \texttt{covglasso}} & & {\footnotesize \algoname} & {\footnotesize \texttt{covglasso}} \\
 \hline
 \hline
 25  		   & & 0.028 & 0.008      & & 0.055 & 0.005 \\
 \hline
 50   		   & & 0.023 & 0.056     & & 0.020 & 0.042 \\
 \hline
 100 		   & & 0.093 & 1.962     & & 0.240 & 1.807 \\
 \hline
 200 		   & & 0.515 & 21.765   & & 1.231 & 19.765 \\
 \hline
 300 		   & & 2.760 & 110.817 & & 4.780 & 76.118 \\
 \hline
\end{tabular}
\end{table}
The results show how our proposed algorithm is generally faster than \texttt{covglasso}, with a computational gain that increases the larger $p$ is.

\subsection{Ridge-regularized covglasso} \label{SCT:simulations.lasso}
In this final simulation study, we consider the general case of sparse estimation of $\Sigma \in \Set_{\succ}$. The main competitor here is covglasso  \citep{lam2009sparsistency}, 
with a parameter $\lambda$ that can be tuned to induce sparsity in the estimates of the covariance. The interest is to compare the behaviour of the covglasso estimator with the most general version of our proposed ridge-regularized covglasso, with a tuning parameter $\kappa$ that induces stability in the estimates across the different sparsity levels controlled by $\lambda$.

As in Section~\ref{SCT:simul.mle}, we fix the number of variables to $p=50$, consider a banded covariance matrix $\Sigma$, with a density varying across three levels ($10\%, 30\%$ and $50\%$), and simulate data with sample sizes  $n \in \{45, 75, 100, 250, 500, 1000\}$.
The parameters $\lambda$ and $\kappa$ for our approach are selected via $5$-fold cross-validation using the maximum theoretical values and the strategy derived in Section~\ref{SCT:model.selection}. When $n>p$,  a comparison is made with covglasso and the same grid of values is chosen for the  tuning parameter $\lambda$.

Figure~\ref{IMG.sim.lasso} reports the average results across 10 simulations. Figure~\ref{IMG.sim.lassoa} shows the same behaviour for $\kappa$ across $n$ that was observed in Figure~\ref{IMG.sim.mlea}, i.e., the need for stabilizing the sample covariance matrix, particularly when $n$ is small. Figure~\ref{IMG.sim.lassob} plots the optimal values of $\lambda$, normalized by their maximum theoretical value when $\kappa=0$. The results, confirmed also by Figure~\ref{IMG.sim.lasso.fullf} in Appendix~\ref{AX:simulations.lasso}, show how the stabilization induced by $\kappa$ allows the ridge-regularized covglasso approach to select denser covariance matrices than covglasso. In particular, this also shows how the optimal value of $\kappa$ depends on characteristics of the model that are unknown prior to inference, and thus that a careful tuning of this parameter is needed.

Figures~\ref{IMG.sim.lassoc} and \ref{IMG.sim.lassod}, and the additional figures reported in Appendix~\ref{AX:simulations.lasso}, further evaluate the quality of the estimators in terms of the accuracy of the estimation and of the recovery of the sparsity pattern. For the former, we consider the root mean square error (RMSE), defined by  $\|\widehat{\Sigma} - \Sigma\|_F/p^2$ (Figures \ref{IMG.sim.lassoc}) and the entropy loss (Figure~\ref{IMG.sim.lasso.fullb} in Appendix~\ref{AX:simulations.lasso}), while for the latter we use the $F_1$ score (Figure \ref{IMG.sim.lassod}) and the true positive rate (Figure \ref{IMG.sim.lasso.fullc}), true negative rate (Figure~\ref{IMG.sim.lasso.fulld}) and positive predictive value (Figure~\ref{IMG.sim.lasso.fulle}). The plots show how the $\widehat{\Sigma}$ matrices estimated by the ridge-regularized covglasso approach have better properties compared to those estimated by covglasso, across all settings and particularly so for small sample sizes. The results on the graph recovery confirm how, for small $n$, our proposed estimator results in denser graphs and further show how these graphs are closer to the true ones. Indeed, compared to covglasso, the graphs recovered by our procedure are associated to a higher true positive rate (Figure \ref{IMG.sim.lasso.fullc}). This overcompensates from the lower positive predictive value (Figure~\ref{IMG.sim.lasso.fulle}), as the harmonic mean of these two measures, which gives the $F_1$ score (Figure \ref{IMG.sim.lassod}), shows an an overall superior performance of the ridge-regularized covglasso estimators.

\newcommand{\locsubfwidth}{0.45\textwidth}
\newcommand{\locmidspace}{0.5cm}
\begin{figure}[!htbp]
\vspace*{-15bp}
\centering
	\begin{subfigure}[t]{\locsubfwidth}
	\centering
	\caption{\label{IMG.sim.lassoa}}
	\includegraphics[width = \linewidth]{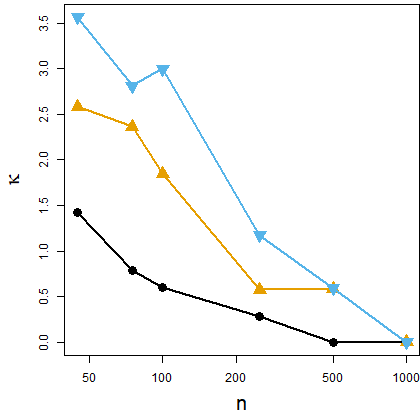}
	\end{subfigure}
	\hspace{\locmidspace}
	\begin{subfigure}[t]{\locsubfwidth}
	\centering
	\caption{\label{IMG.sim.lassob}}
	\includegraphics[width = \linewidth]{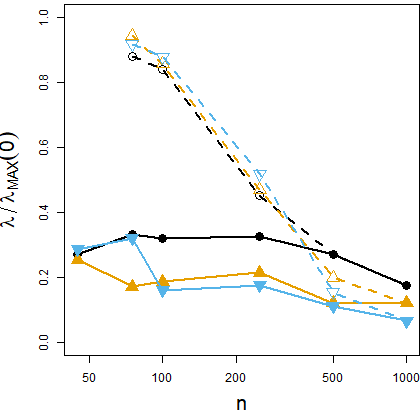}
	\end{subfigure}
	
	\vspace{-0.3cm}
	\begin{subfigure}[t]{\locsubfwidth}
	\centering
	\caption{\label{IMG.sim.lassoc}}
	\includegraphics[width = \linewidth]{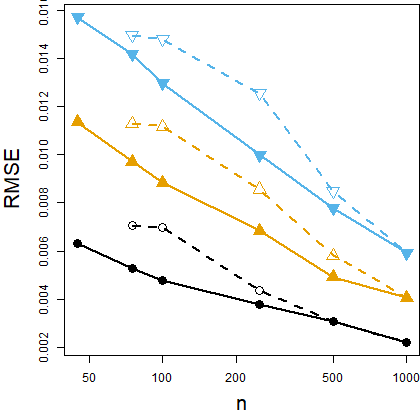}
	\end{subfigure}
	\hspace{\locmidspace}
	\vspace{-2bp}
	\begin{subfigure}[t]{\locsubfwidth}
	\centering
	\caption{\label{IMG.sim.lassod}}
  \includegraphics[width = \linewidth]{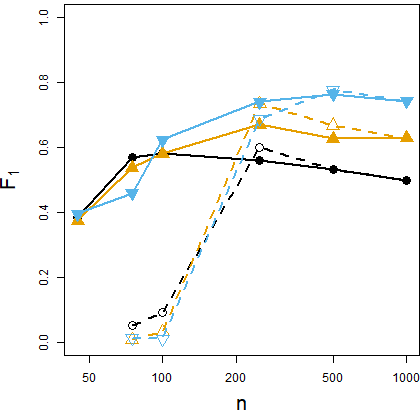}
	\end{subfigure}
	\vspace{2bp}
	\includegraphics[scale=0.6]{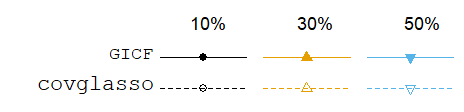}
\caption{Average results from 10 simulations with sample size $n \in \{45, 75, 100, 250, 500, 1000\}$, $p=50$ and banded covariance matrices $\Sigma$ associated to graphs $\G$ with varying density levels (10\%, 30\%, 50\%). Comparison between \algoname{} ($\kappa \geq 0,\, \geq 0$) and covglasso ($\kappa = 0,\, \lambda \geq 0$) shows how  (a) the optimal $\kappa$ decreases with $n$, (b) the optimal $\lambda/\lambda_{MAX}(0)$ is smaller for \algoname{}, suggesting denser estimated graphs, (c)  $\widehat{\Sigma}$ of \algoname{} has a smaller RMSE than covglasso, i.e., a more accurate estimation of $\Sigma$ (d) \algoname{} has a higher $F_1$ score than covglasso, i.e., a better recovery of the true graphs $\G$.}\label{IMG.sim.lasso}
\end{figure}

\section{Analysing sonar data} \label{SEC:application}
In this section, we show an illustration of our approach to the sonar dataset  from the UCI Machine Learning Repository (\url{https://doi.org/10.24432/C5T01Q}).
The data contain sonar signals summarised into $60$ continuous features in the range $[0, 1]$. Each feature represents a measure of the energy within a specific frequency band. The signals are measured on $208$ spectra, of which $97$ are obtained from a rock sample and $111$ from a metal cylinder. Each observation is measured at a different angle on the same object. Following previous analyses, we assume independence across the spectra \citep{rothman2010cholesky}. Further details about the pre-processing of the signals can be found in \citet{gorman1988learned}.

\begin{figure}[!tbp]
\centering
\begin{tabular}{m{1.2cm} m{0.25\linewidth} m{0.25\linewidth} m{0.25\linewidth} p{1.5cm}}
 & \multicolumn{1}{c}{Sample} & \multicolumn{1}{c}{Banded} & \multicolumn{1}{c}{Unstructured} & \\
\rotatebox[origin=c]{0}{Rock} & \includegraphics[width = \linewidth]{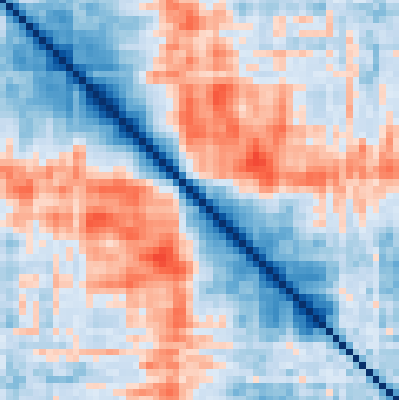} & \includegraphics[width = \linewidth]{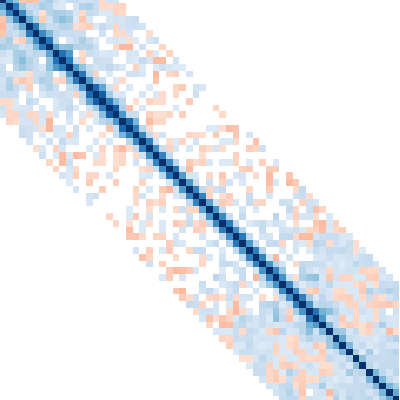} & \includegraphics[width = \linewidth]{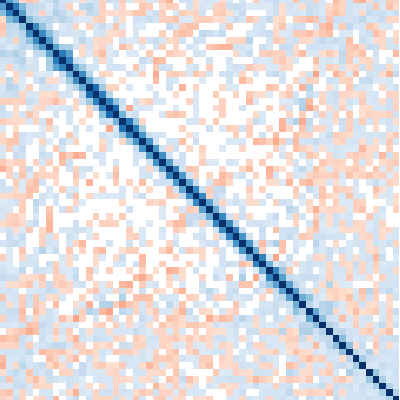} & \multirow{2}{*}{\includegraphics[width = 0.6\linewidth]{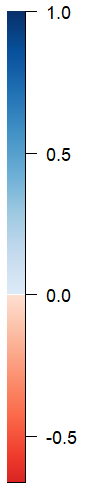}}\\
\rotatebox[origin=c]{0}{Metal} & \includegraphics[width = \linewidth]{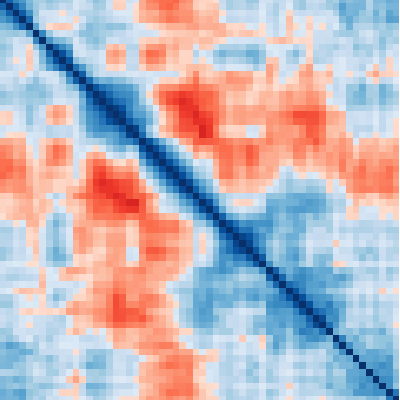} & \includegraphics[width = \linewidth]{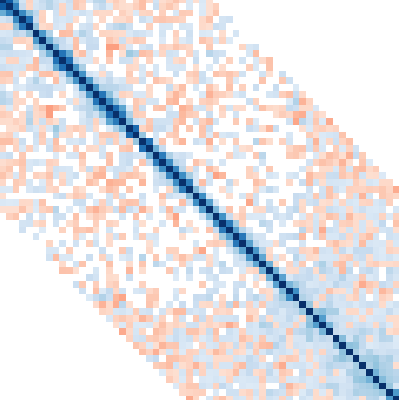} & \includegraphics[width = \linewidth]{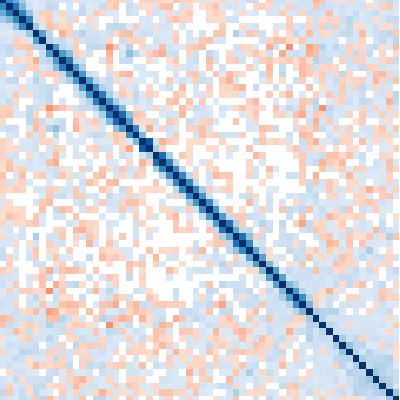} &
\end{tabular}
\caption{Analsyis of sonar data: heatmaps of the sample (left) and ridge-regularized covglasso correlation matrices for the rock  (top) and  metal (bottom) sonar spectra. For our method, we consider both a prespecified banded structure of zeros (middle) and no prior specification (right). }\label{IMG.real.data.estimates}
\end{figure}
Figure~\ref{IMG.real.data.estimates} (left) shows the sample correlation matrices for the metal and rock samples. The nature of the features, sorted by frequency bands, explains the specific structure that is visible on this plot, with correlations that decay away from the diagonal. This has motivated the use of banded estimators of $\Sigma$ on these data \citep{rothman2010cholesky}. In particular, in their paper, \cite{rothman2010cholesky} identify a banded structure with $31$ bands for metal and $17$ bands for rock.
In this paper, we relax this assumption using our proposed approach. In particular, we consider two cases of ridge-regularized covglasso. In a first scenario, we impose the zeros in $\G$ induced by the banded structures of \citet{rothman2010cholesky} but allow the tuning parameter $\lambda$ to refine this structure further, i.e., to detect additional sparsity within the banded structures. In a second scenario, we do not impose any a priori knowledge of $\G$. In both cases, we stabilize the estimators by tuning the parameter $\kappa$, and we select both $\kappa$ and $\lambda$ via $10$-fold cross-validation.

The middle and right plots of Figure~\ref{IMG.real.data.estimates} show the estimated correlation matrices using the two approaches, respectively. For ease of visualization, we plot the MLEs of the correlation matrices constrained to the sparsity patterns $\G$ detected by the method.
The results show how the banded estimates identify some additional zeros within the bands. These are also confirmed by the unstructured estimates. The latter, however, do not show any pronounced banded pattern, beside the rapidly decaying correlations from the main off-diagonal, which are common to all the estimates presented in Figure~\ref{IMG.real.data.estimates}. Overall, these results suggest that imposing a banded structure on these data may be a too strong constraint.

As a further validation, and following \citet{rothman2010cholesky}, we employ our ridge-regularized covglasso estimators within a quadratic discriminant analysis, in order to classify spectra into rock and metal based on the sonar signals.  In particular, we classify each observation $\vctr{x}$ based on
\[\argmax_{j \in \{\text{rock},\text{metal}\}}\{\log|\widehat{\Sigma}_j^{-1}| - (\vctr{x} - \hat{\vctr{\mu}}_j)^\top\widehat{\Sigma}_j^{-1}(\vctr{x} - \hat{\vctr{\mu}}_j) + 2\log{\hat{\pi}_j}\},\] where $\widehat{\Sigma}_j, \hat{\vctr{\mu}}_j, \hat{\pi}_j$ are, respectively, the estimated covariance matrix, the sample mean and the proportion of observations belonging to class $j$. We evaluate the performance of the approaches using 5-fold cross-validation, where the tuning parameters $\lambda$ and $\kappa$ for the estimation of $\widehat{\Sigma}_j$ are selected from each leave-one-out set via $10$-fold cross-validation.

Table~\ref{TAB:sonar.qda} summarises the results for both the banded and unstructured cases discussed above. For completeness, together with the most general \algoname{} version discussed above, we also include the cases when $\kappa$ and/or $\lambda$ are set to zero, which correspond to, or get close to, existing methods as discussed in the simulation study.
The results show how, in general, the introduction of the tuning parameter $\kappa$ improves the classification error with respect to the MLE. On the other hand, the sole introduction of the tuning parameter $\lambda$ does not improve the classification error, but it is beneficial if used jointly with $\kappa$, and the highest improvement is achieved when both regularization parameters are employed. Furthermore, we find that the unstructured estimates perform significantly better than their banded counterparts.

\begin{table}[!tbp]
\caption{The test errors (in percentage) for the quadratic discriminant analysis performed on sonar data to discriminate metal and rock spectra.}
\label{TAB:sonar.qda}
\begin{center}
\begin{tabular}{|c|*{4}{c|}}
\hline
\multicolumn{1}{|c}{\multirow{2}{*}{}} & \multicolumn{4}{|c|}{Parameter settings} \\
\cline{2-5}
& $\lambda=\kappa=0$ & $\lambda \geq 0, \kappa=0$ & $\lambda = 0, \kappa\geq 0$ & $\lambda \geq 0, \kappa \geq 0$ \\
\hline\hline\setcounter{rowcntr}{1}
Banded & 25.83 & 34.65 & 25.44 & 21.09 \\
\hline
Unstructured & 23.08 & 28.68 & 19.72 & 17.24\\
\hline
\end{tabular}
\end{center}
\end{table}

\section{Conclusions}\label{SEC:conclusions}

A large number of different methods have been introduced in the literature for the estimation of covariance matrices in high dimensions, which is of interest in a number of applications. The task is made challenging by the poor behaviour of the sample covariance matrix in the high dimensional setting. In this paper, we have proposed a regularized Gaussian loglikelihood approach which is able to accommodate both a sparsity-inducing $\ell_{1}$-penalty and a stability-inducing ridge penalty. We have derived a computationally efficient ICF-type algorithm for the optimization of the regularized loglikelihood and have implemented a cross-validated loglikelihood procedure for the selection of the tuning parameters. The implementation allows sparse estimation also in high dimensional settings and can naturally account for a partially known sparsity pattern in the covariance matrix.

A simulation study showed a good performance of the proposed approach, both in terms of identification of the sparsity structure and properties of the estimator of the unknown covariance. This is particularly the case for small sample sizes, though the ridge regularization is found to play an important role also when $n$ is not much larger than $p$. Both the simulation study and the illustration on sonar data show the advantages of a joint application of both penalties, and therefore the need for a principled approach for tuning the two penalties simultaneously. In particular, the simulation study shows how the ridge-regularization term allows the procedure to tune the underlying sparsity level better than when using the $\ell_{1}$-penalty alone, leading to graphs with a higher $F_{1}$ score.  The unifying procedure that we have proposed means that either of the regularization terms can also be selected to be zero, leading to special cases of existing methods. These can now be applied also in the presence of a partially known sparsity pattern, in the case of \citet{warton2008penalized}, and/or a high dimensional setting, in the case of \citet{wang2014coordinate}.  

\section*{Code and data availability}
The  \texttt{R} package \packagename{} and the \texttt{R} scripts for reproducing the simulations of Section \ref{SCT:simulations} and the illustration of Section \ref{SEC:application} can be found at \url{https://github.com/luca-cibinel/gicf}.

\section*{Acknowledgments}
Alberto Roverato and Veronica Vinciotti acknowledge funding from the the European Union - Next Generation EU, Mission 4 Component 2 - CUP C53D23002580006 (MUR-PRIN grant 2022SMNNKY).
Luca Cibinel was supported by the National Cancer Institute of the National Institutes of Health (U24CA180996).

\bibliographystyle{chicago}
\bibliography{hcovglasso-ref}




\newpage
\begin{appendices}

\setcounter{equation}{0}
\setcounter{figure}{0}
\setcounter{table}{0}
\setcounter{theorem}{0}
\setcounter{example}{0}
\setcounter{definition}{0}
\setcounter{algorithm}{0}
\newcommand{\axpref}{A}
\renewcommand{\theequation}{\axpref\arabic{equation}}
\renewcommand{\thefigure}{\axpref\arabic{figure}}
\renewcommand{\thetable}{\axpref\arabic{table}}
\renewcommand{\thelemma}{\axpref\arabic{lemma}}
\renewcommand{\theprop}{\axpref\arabic{prop}}
\renewcommand{\thetheorem}{\axpref\arabic{theorem}}
\renewcommand{\thecorollary}{\axpref\arabic{corollary}}
\renewcommand{\theexample}{\axpref\arabic{example}}
\renewcommand{\thedefinition}{\axpref\arabic{definition}}
\renewcommand{\thealgorithm}{\axpref\arabic{algorithm}}

\section{Pseudocodes}\label{APP:pseudocodes}
We provide here, in Algorithms~\ref{ALG.algo.detailed} and \ref{ALG.cd.lin.detailed}, the pseudocodes for the implementation of the \algoname{} algorithm described in Section~\ref{SEC:hybrid.procedure}.
{
\begin{algorithm}
	\caption{\algoname{} algorithm for the solution of (\ref{EQ:shrinkage.covglasso.estimator})}\label{ALG.algo.detailed}
	\begin{algorithmic}[1]
	\Require $\G = (V, E)$,\; $\widehat{\Sigma}^{(0)}$,\; $Y$,\;  $\lambda$ and $\kappa$
	\State $r\gets 0$
	\State $n \gets \text{number of rows of}\ Y$
	\State $Y^{(\kappa)} \gets (Y^\top,\, \sqrt{n\kappa}I_p)^\top$
	\Repeat
		\State $r \gets r + 1$
		\State $\widehat{\Sigma}^{(r)} \gets \hat{\Sigma}^{(r-1)}$
		\ForAll{$i \in V$}
            \If{$\bnd{i}\neq \emptyset$}\Comment{{\footnotesize $\bnd{i}$ is computed with respect to $\G$}}
			\State $\widehat{Z}^{(r, \kappa)}_{\bnd{i}} \gets Y^{(\kappa)}_{-i}\left[\left(\widehat{\Sigma}^{(r)}_{-i,-i}\right)^{-1}\right]_{-i,\bnd{i}}$
			\State $\hat{\tau}^{(r)}_i,\; \sigmavh_{\bnd{i},i}^{(r)} \gets$ \Call{\algoname{}step}{$\sigmavh_{\bnd{i},i}^{(r)},\; \vctr{y}^{(\kappa)}_i,\; \widehat{Z}^{(r, \kappa)}_{\bnd{i}},\; \lambda,\; n$}\label{LINE.call.step}\Comment{{\footnotesize see Algorithm~\ref{ALG.cd.lin.detailed}}}
            \Else
                \State $\hat{\tau}^{(r)}_{i} \gets \frac{1}{n}\left(\vctr{y}_i^{(\kappa)}\right)^\top\vctr{y}_i^{(\kappa)}$
            \EndIf
			\State $\sigmavh_{-i \setminus \bnd{i},i}^{(r)} \gets \vctr{0}$\label{LINE.set.sigma.zeros}
			\State $\hat{\sigma}_{ii}^{(r)} \gets \hat{\tau}_i^{(r)} + \left(\sigmavh^{(r)}_{-i,i}\right)^\top\left(\widehat{\Sigma}^{(r)}_{-i,-i}\right)^{-1}\sigmavh^{(r)}_{-i,i}$
\label{LINE.update.sigmaii}
		\EndFor
	\Until{\textbf{convergence is reached}}
	\State \textbf{return} $\widehat{\Sigma}^{(r)}$
	\end{algorithmic}
\end{algorithm}

\begin{algorithm}
	\caption{Coordinate-descent algorithm for the optimization of (\ref{EQ:lin.reg.touse})}\label{ALG.cd.lin.detailed}
	\begin{algorithmic}[1]
	\Procedure{\algoname{}step}{$\hat{\vctr{\beta}}^{(0)},\; \vctr{y},\; Z,\; \lambda,\; n$}
	\State $r\gets 0$
	\State $d \gets \text{number of elements of}\ \hat{\vctr{\beta}}^{(0)}$
	\Repeat
		\State $r \gets r + 1$
		\State $\hat{\tau}^{(r)} \gets \frac{1}{n}\|\vctr{y} - Z\hat{\vctr{\beta}}^{(r-1)}\|_2^2$\label{LINE.cd.assign.tau}
		\State $\hat{\vctr{\beta}}^{(r)}\gets \hat{\vctr{\beta}}^{(r-1)}$
		\For{$j = 1, \dots, d$} 
			\State $B_{j} \gets \frac{1}{n\hat{\tau}^{(r)}}\left(\left[Z^\top \vctr{y}\right]_{j} - \sum_{\substack{m = 1\\m \neq j}}^d\left[Z^\top Z\right]_{jm}\hat{\beta}^{(r)}_{m}\right)$\label{LINE.compute.B}
			\State $C_{j} \gets \frac{1}{n\hat{\tau}^{(r)}}\left[Z^\top Z\right]_{jj}$
			\State $\hat{\beta}^{(r)}_{j} \gets \frac{1}{C_{j}}\mathcal{S}_{\lambda}\left(B_{j}\right)$\label{LINE.soft.thr}
		\EndFor
	\Until{\textbf{convergence is reached} }
	\State \textbf{return} $\hat{\tau}^{(r)}$ and $\hat{\vctr{\beta}}^{(r)}$
	\EndProcedure
	\end{algorithmic}
\end{algorithm}
}

\section{Preliminary lemmas}\label{AX:bounds}
In this section we present some results concerning the behaviour of Algorithms~\ref{ALG.algo.detailed} and \ref{ALG.cd.lin.detailed}, which will be exploited in the proofs of Theorems~\ref{THM.diag.S} and \ref{THM:kappa_max}.

We first provide a sufficient condition on the value of $\lambda$ for Algorithm~\ref{ALG.cd.lin.detailed} to converge to a zero vector of regression coefficients.
\begin{lemma}\label{THM:max.lambda.alg.2}
Consider Algorithm~\ref{ALG.cd.lin.detailed}. If it holds that
\begin{align}\label{EQN:condition.lambda.A2}
\lambda \geq \frac{1}{ns^{2}}\max_{j \in \{1,\ldots, d\}}{\left| \left[Z^\top \vctr{y}\right]_j\right|}
\end{align}
where $s^{2} = \frac{1}{n}\vctr{y}^\top\vctr{y}$, and the starting point $\hat{\vctr{\beta}}^{(0)}$ is set to $\vctr{0}$,  then
the algorithm converges, in one iteration, to $\hat{\tau}=s^{2}$ and $\hat{\vctr{\beta}}=\vctr{0}$.
\end{lemma}
\begin{proof}
In order to prove the result it is sufficient to set $\hat{\vctr{\beta}}^{(0)}=\vctr{0}$ and to show that, if (\ref{EQN:condition.lambda.A2}) holds true, then at the beginning of
the second iteration of the repeat-cycle of Algorithm~\ref{ALG.cd.lin.detailed} one has $\hat{\tau}^{(1)}=s^{2}$ and  $\hat{\vctr{\beta}}^{(1)}=\hat{\vctr{\beta}}^{(0)}=\vctr{0}$. As consequence, the second, as well as every subsequent iteration,  will be executed  from the same staring point and will thus produce the same output, so that $\hat{\tau}^{(r)}=s^{2}$ and  $\hat{\vctr{\beta}}^{(r)}=\vctr{0}$, for every $r\geq 1$.

Because $\hat{\vctr{\beta}}^{(0)}=\vctr{0}$, the first iteration of the  repeat-cycle of the algorithm sets $\hat{\tau}^{(1)}=s^{2}$ in  line~\ref{LINE.cd.assign.tau}. Next, $\hat{\vctr{\beta}}^{(1)}$ is set to $\hat{\vctr{\beta}}^{(0)}=\vctr{0}$, and then the for-cycle updates, in turn, every entry $\hat{\beta}^{(1)}_{j}$ of $\hat{\vctr{\beta}}^{(1)}$, for $j=1,\ldots, d$. For $j=1$, it follows from
line~\ref{LINE.soft.thr} of the algorithm that $\hat{\beta}^{(1)}_{1}$ does not change its value, that is
$\hat{\beta}^{(1)}_{1}=0$, if and only if
$\mathcal{S}_{\lambda}\left(B_{1}\right)=0$, that is if and only if $\lambda\geq |B_{1}|$, and this inequality holds true by (\ref{EQN:condition.lambda.A2}), because we can see from line~\ref{LINE.compute.B} of the algorithm that,
\begin{align*}
|B_{1}|
=\left|\frac{1}{n\hat{\tau}^{(1)}}\left(\left[Z^\top \vctr{y}\right]_{1} - \sum_{m = 2}^{d}\left[Z^\top Z\right]_{1m}\hat{\beta}^{(1)}_{m}\right)\right|
=\frac{1}{ns^{2}}\left| \left[Z^\top \vctr{y}\right]_{1}\right|
\leq \frac{1}{ns^{2}}\max_{j \in \{1,\ldots, d\}}{\left| \left[Z^\top \vctr{y}\right]_{j}\right|}.
\end{align*}
We can then use the same argument for every $j=2,\dots, d$ to show that  the second iteration of the repeat-cycle  will be executed from $\hat{\tau}^{(1)}=s^{2}$ and $\hat{\vctr{\beta}}^{(1)}=\vctr{0}$, as required.
\end{proof}
Algorithm~\ref{ALG.cd.lin.detailed} is called by Algorithm~\ref{ALG.algo.detailed}, and we now restate Lemma~\ref{THM:max.lambda.alg.2} above with respect to such call.
\begin{lemma}\label{THM:call.to.alg.A2}
Consider the call to Algorithm~\ref{ALG.cd.lin.detailed} in line~\ref{LINE.call.step} of Algorithm~\ref{ALG.algo.detailed}, that is,
\begin{center}
$\hat{\tau}^{(r)}_{i},\; \sigmavh_{\bnd{i},i}^{(r)} \gets$ \Call{\algoname{}step}{$\sigmavh_{\bnd{i},i}^{(r)},\; \vctr{y}^{(\kappa)}_i,\; \widehat{Z}^{(r, \kappa)}_{\bnd{i}},\; \lambda,\; n$},
\end{center}
for $i\in V$. If the input values are such that
$\sigmavh_{\bnd{i},i}^{(r)}=\vctr{0}$ and
\begin{align}\label{EQN:condition.lambda.call}
\lambda \geq \frac{1}{ns^{(\kappa)}_{ii}}\max_{j \in \bnd{i}}{\left| \left[(\widehat{Z}_{\bnd{i}}^{(r, \kappa)})^\top \vctr{y}_i^{(\kappa)}\right]_j\right|}
\end{align}
where $s^{(\kappa)}_{ii} = \frac{1}{n}\left(\vctr{y}_i^{(\kappa)}\right)^\top\vctr{y}_i^{(\kappa)}$, then the output values are
$\sigmavh_{\bnd{i},i}^{(r)}=\vctr{0}$ and $\hat{\tau}^{(r)}_{i}=s^{(\kappa)}_{ii}$.
\end{lemma}

\begin{proof}
This follows immediately from Lemma~\ref{THM:max.lambda.alg.2} by setting the input values
$\hat{\vctr{\beta}}^{(0)}=\sigmavh_{\bnd{i},i}^{(r)}$ $\vctr{y}=\vctr{y}^{(\kappa)}_{i}$  and $Z=\widehat{Z}^{(r, \kappa)}_{\bnd{i}}$, so that the output values are $\hat{\tau}=\hat{\tau}^{(r)}_{i}$  and $\hat{\vctr{\beta}}=\sigmavh_{\bnd{i},i}^{(r)}$.
\end{proof}
Finally, we provide a sufficient condition for the \algoname{} algorithm to provide a diagonal solution.
\begin{lemma}[Sufficient condition for a diagonal solution]\label{LEMMA.suff.l.max}
Consider Algorithm \ref{ALG.algo.detailed} (\algoname{}). If it holds that
\begin{equation}\label{EQ.one.step.diagonal}
\lambda \geq \max_{\{i,j\} \in \G}\frac{|s_{ij}|}{(s_{ii}+\kappa) (s_{jj}+\kappa)},
\end{equation}
and the starting point is set to $\widehat{\Sigma}^{(0)} = \diag{S^{(\kappa)}}$, then $\widehat{\Sigma}^{(r)} = \diag{S^{(\kappa)}}$ for each $r\geq 1$.
\end{lemma}
\begin{proof}

In order to prove the result it is sufficient to set $\widehat{\Sigma}^{(0)} = \diag{S^{(\kappa)}}$ and to show that, if (\ref{EQ.one.step.diagonal}) is satisfied, then at the beginning of
the second iteration of the repeat-cycle of Algorithm~\ref{ALG.algo.detailed} one has $\widehat{\Sigma}^{(1)} =\widehat{\Sigma}^{(0)} = \diag{S^{(\kappa)}}$. As consequence, the second, as well as every subsequent iteration,  will be executed from the same staring point and will thus produce the same output, so that $\widehat{\Sigma}^{(r)} = \diag{S^{(\kappa)}}$, for every $r\geq 1$.

Consider the first iteration of the repeat-cycle of the algorithm. Firstly,
$\widehat{\Sigma}^{(1)}$ is set to $\widehat{\Sigma}^{(0)} = \diag{S^{(\kappa)}}$, and then a for-cycle updates, for $i=1,\ldots,p$, the scalar
$\hat{\sigma}^{(1)}_{ii}$ and the vector $\sigmavh_{-i,i}^{(1)}$, which form the $i$-th row and column of $\widehat{\Sigma}^{(1)}$.

We first show that at the end of the first iteration of the for-cycle, i.e. $i=1$, one has $\hat{\tau}^{(1)}_{1}=s_{11}^{(\kappa)}$ and $\sigmavh_{-1,1}^{(1)}=\vctr{0}$, an then, at  line~\ref{LINE.update.sigmaii}, also  $\hat{\sigma}^{(1)}_{11}=\hat{\tau}^{(1)}_{1}=s_{11}^{(\kappa)}$. The result is trivially true if $\bnd{i}=\emptyset$ and, furthermore, because  line~\ref{LINE.set.sigma.zeros} of the algorithm sets the entries
$\sigmavh_{-1 \setminus \bnd{1},1}^{(1)}$ to $\vctr{0}$, then if $\bnd{i}\neq \emptyset$ is sufficient to show that the call
\begin{center}
$\hat{\tau}^{(1)}_{1},\; \sigmavh_{\bnd{1},1}^{(1)} \gets$ \Call{\algoname{}step}{$\sigmavh_{\bnd{1},1}^{(1)},\; \vctr{y}^{(\kappa)}_1,\; \widehat{Z}^{(1, \kappa)}_{\bnd{i}},\; \lambda,\; n$},
\end{center}
returns the values $\hat{\tau}^{(1)}_{1}=s^{(\kappa)}_{11}$ and $\sigmavh_{\bnd{1},1}^{(1)}=\vctr{0}$. This follows from Lemma~\ref{THM:call.to.alg.A2} because, in the above call, the input $\sigmavh_{\bnd{1},1}^{(1)}$ is the zero vector and, furthermore, condition (\ref{EQN:condition.lambda.call}) is satisfied. To formally see the latter,  we have to show that (\ref{EQ.one.step.diagonal}) implies
\begin{align}\label{EQN:condi.lambda.call.first.variable}
\lambda \geq\max_{j \in \bnd{1}}{\left|\left[\frac{1}{ns^{(\kappa)}_{11}} (\widehat{Z}_{\bnd{1}}^{(1, \kappa)})^\top \vctr{y}_1^{(\kappa)}\right]_j\right|},
\end{align}
that is (\ref{EQN:condition.lambda.call}) when $r=1$ and $i=1$. This can be done by noticing that, in (\ref{EQN:condi.lambda.call.first.variable}),
\begin{align*}
\frac{1}{ns^{(\kappa)}_{11}}(\widehat{Z}_{\bnd{1}}^{(1, \kappa)})^\top \vctr{y}_1^{(\kappa)}
&=\frac{1}{ns^{(\kappa)}_{11}}\left[\left(\diag{S^{(\kappa)}}_{-1,-1}\right)^{-1}\right]_{\bnd{1},-1}(Y^{(\kappa)}_{-1})^{\top}\vctr{y}_1^{(\kappa)}\\
&=\frac{1}{s^{(\kappa)}_{11}}\left[\left(\diag{S^{(\kappa)}}_{-1,-1}\right)^{-1}\right]_{\bnd{1},-1}S^{(\kappa)}_{-1,1}\\
&=\left[\frac{s_{1j}^{(\kappa)}}{s^{(\kappa)}_{11}s^{(\kappa)}_{jj}}\right]_{j\in \bnd{1}}\\
&=\left[\frac{s_{1j}}{(s_{11}+\kappa)(s_{jj}+\kappa)}\right]_{j\in \bnd{1}}.
\end{align*}
Hence, if we recall that $j\in \bnd{1}$ if and only if $\{1, j\}\in \G$, we can write (\ref{EQN:condi.lambda.call.first.variable}) as,
\begin{align*}
\lambda
\geq \max_{\{1,j\} \in \G}{\frac{|s_{1j}|}{(s_{11}+\kappa)(s_{jj}+\kappa)}}
\end{align*}
that is implied by (\ref{EQ.one.step.diagonal}) because,
\begin{align*}
\max_{\{1,j\} \in \G}{\frac{|s_{1j}|}{(s_{11}+\kappa)(s_{jj}+\kappa)}}
\leq
\max_{\{i,j\} \in \G}\frac{|s_{ij}|}{(s_{ii}+\kappa) (s_{jj}+\kappa)}.
\end{align*}
We can then use the same argument for every $i=2,\dots, p$ to show that  the second iteration of the repeat-cycle  will be executed from
$\widehat{\Sigma}^{(1)} = \diag{S^{(\kappa)}}$, as required.
\end{proof}

\section{Proofs}\label{AX:proofs}

\subsection{Proof of Proposition~\ref{THM:shrinkage.transform}}\label{PROOF:THM:shrinkage.transform}

In order to prove (\ref{EQ:shrinkage.1}), it is sufficient to show that $\ell_{Y^{(\kappa)},n}(\Sigma)=\ell_{Y}(\Sigma) - \kappa\,\|\diag{\Sigma^{-1}}\|_{1}$, and the latter equality holds true because it follows from (\ref{EQ:log.lik}) that,
\begin{align}
\nonumber
\ell_{Y^{(\kappa)},n}(\Sigma)
&= -\log{|\Sigma|} -\frac{1}{n}\tr(\Sigma^{-1}Y^{(\kappa)\top}Y^{(\kappa)})\\
\label{EQN:ridge-reg-connection}
&= -\log{|\Sigma|} -\frac{1}{n}\tr\{\Sigma^{-1}(Y^{\top}Y+n\kappa I_{p})\}\\
\nonumber
&= -\log{|\Sigma|} -\frac{1}{n}\tr(\Sigma^{-1}Y^{\top}Y)-\kappa\tr( \Sigma^{-1})\\
\nonumber
&= \ell_{Y}(\Sigma) - \kappa\,\|\diag{\Sigma^{-1}}\|_{1}.
\end{align}
Similarly, the equality (\ref{EQ:shrinkage.2}) holds true if and only if $\ell_{Y^{(\kappa)},n}(\Sigma)=-\log{|\Sigma|} -\tr(\Sigma^{-1}S^{(\kappa)})$ and this follows immediately from (\ref{EQN:ridge-reg-connection}).

\subsection{Proof of Corollary~\ref{THM:final.loglik.decomposition}}\label{PROOF:THM:final.loglik.decomposition}
The result follows from (\ref{EQ:shrinkage.1}) by (i) applying the decomposition (\ref{EQ:log.lik.split}) to $\ell_{Y^{(\kappa)},n}(\Sigma)$, (ii) by noticing that
\begin{align*}
\lambda\,\|\Sigma-\diag{\Sigma}\|_1
=\lambda\|\Sigmaii - \diag{\Sigmaii}\|_1
+ 2\lambda\|\sigmaii\|_1
\end{align*}
and, finally, (iii) by recalling that $\Sigma\in\mathcal{S}_{\succ}(\G)$ implies that $\sigmaii^\top = (\sigmabdi^\top, \vctr{0}^\top)$ so that
$Z^{(\kappa)}_{-i}\sigmaii=Z^{(\kappa)}_{\bnd{i}}\sigmabdi$ and
$\|\sigmaii\|_1=\|\sigmabdi\|_1$.

\subsection{Proof of Theorem~\ref{THM.diag.S}}\label{AX:proof.diag.S}
If \algoname{} is invoked with $\widehat{\Sigma}^{(0)} = \diag{S^{(\kappa)}}$ and $\lambda \geq \lambda_{MAX}(\kappa)$, then by Lemma \ref{LEMMA.suff.l.max} the sequence of intermediate estimates $\mathcal{E}$ is identically equal to $\diag{S^{(\kappa)}}$. In particular, $\diag{S^{(\kappa)}}\in\mathcal{S}_{\succ}(\G)$ is an accumulation point of $\mathcal{E}$ and, by Proposition \ref{PROP.convergence}, also a stationary point of  $\ell_{Y}(\Sigma) - \Ps_{\lambda,\kappa}(\Sigma)$ in (\ref{EQ:shrinkage.covglasso.estimator}).

\subsection{Proof of Theorem~\ref{THM:kappa_max}}\label{AX:proof.kappa.max}

We first remark that the assumption $\lambda_{MAX}(\kappa)\neq 0$ implies both that $\G = (V, E)$ is such that $E \neq \emptyset$, and that $S$ is such that $\max_{\{i,j\} \in \G}|s_{ij}|> 0$. Then, for every $\{i, j\} \in \G$, and for $\kappa \geq 0$, we set
\begin{align*}
\lambda_{ij}(\kappa)
= \frac{|s_{ij}|}{(s_{ii} + \kappa)(s_{jj} + \kappa)}
= \frac{|s_{ij}|}{\kappa^2 + (s_{ii} + s_{jj})\kappa + s_{ii}s_{jj}},
\end{align*}
so that
\begin{align*}
\lambda_{MAX}(\kappa) = \max_{\{i,j\} \in \G}\lambda_{ij}(\kappa),
\end{align*}
with $\lambda_{MAX}(\kappa)>0$.
It is also useful to recall that, for every $\kappa\geq 0$ there must exists at least one edge in $\G$, that we denote by $\{i^{\kappa},j^{\kappa}\}$,
such that $\lambda_{MAX}(\kappa)=\lambda_{i^{\kappa}j^{\kappa}}(\kappa)$, that for $\kappa=0$ becomes  $\lambda_{MAX}(0)=\lambda_{i^{0}j^{0}}(0)$.

It is straightforward to see that, for every $\{i,j\} \in \G$ such that $s_{ij}\neq 0$, it holds that $\lambda_{ij}(\kappa)$ is a continuous, positive and strictly decreasing function of $\kappa$, with the property that $\lim_{\kappa \to \infty}\lambda_{ij}(\kappa) = 0$. Hence, the same must hold for $\lambda_{MAX}(\kappa)$ and, therefore, $\lambda_{MAX}(\kappa)\leq \lambda_{MAX}(0)$ for every $\kappa\geq 0$.

Consider the edges $\{i,j\}\in \G$ for which $s_{ij}\neq 0$, so that $\lambda_{ij}(\kappa)$ is strictly positive.
By elementary algebraic manipulation, one can see that, for every $\kappa \geq 0$ and $\lambda>0$, it holds that,
\begin{equation}\label{EQ.equiv.ineq.l.max}
\lambda \leq \lambda_{ij}(\kappa)
\quad\Longleftrightarrow\quad
\kappa^2 + (s_{ii} + s_{jj})\kappa - g_{ij}(\lambda) \leq 0,
\end{equation}
where,
\begin{align*}
g_{ij}(\lambda)=\frac{|s_{ij}|}{\lambda} - s_{ii}s_{jj}.
\end{align*}
The relationship (\ref{EQ.equiv.ineq.l.max}) has some useful implications. Consider the subset of $E$ given by
\begin{align*}
E_{\lambda}=\big\{\{i,j\} \in \G: g_{ij}(\lambda) \geq 0\big\}.
\end{align*}
Firstly, (\ref{EQ.equiv.ineq.l.max}) makes it explicit that the two inequalities are never satisfied when $g_{ij}(\lambda)<0$. As a consequence, $\lambda \leq \lambda_{ij}(\kappa)$ implies that $g_{ij}(\lambda)\geq 0$. Hence, for every $\lambda\leq \lambda_{MAX}(\kappa)=\lambda_{i^{\kappa}j^{\kappa}}(\kappa)$ it holds that $g_{i^{\kappa}j^{\kappa}}(\lambda)\geq 0$. For $\kappa=0$ this means that
$g_{i^{0}j^{0}}(\lambda)\geq 0$ for every $\lambda\leq \lambda_{MAX}(0)= \lambda_{i^{0}j^{0}}(0)$. In other words, $\{i^{\kappa},j^{\kappa}\}\in E_{\lambda}$ for every $\lambda\leq \lambda_{MAX}(\kappa)$
and, furthermore, because by assumption $\lambda\leq \lambda_{MAX}(0)$, then the edge $\{i^{0},j^{0}\}$ is always contained in $E_{\lambda}$, that is therefore never empty. Finally, it can be easily computed that, when $g_{ij}(\lambda) \geq 0$,  the two inequalities in (\ref{EQ.equiv.ineq.l.max}) are satisfied if and only if
\begin{align}\label{EQN:solution.kappa}
\kappa\leq \kappa_{ij}(\lambda)
\quad\mbox{with}\quad
\kappa_{ij}(\lambda)=\sqrt{\frac{1}{4}(s_{ii} + s_{jj})^2 + g_{ij}(\lambda)} - \frac{1}{2}(s_{ii} + s_{jj}).
\end{align}

We can now show that $\lambda \leq \lambda_{MAX}(\kappa)$ implies $\kappa \leq \kappa_{MAX}(\lambda)$. This follows immediately from the points shown above because $\lambda \leq \lambda_{MAX}(\kappa)$ is equivalent to $\lambda\leq \lambda_{i^{\kappa}j^{\kappa}}(\kappa)$ and, thus, by (\ref{EQN:solution.kappa}), we have
\begin{align*}
\kappa \leq \kappa_{i^{\kappa}j^{\kappa}}(\lambda) \leq \kappa_{MAX}(\lambda);
\end{align*}
recall that, $\kappa_{MAX}(\lambda)$ is defined as the largest $\kappa_{ij}(\lambda)$ such that
$\{i,j\} \in E_{\lambda}$, and we have shown that $\{i^{\kappa}, j^{\kappa}\}\in E_{\lambda}$.

We now turn to the reverse implication and show that $\kappa \leq \kappa_{MAX}(\lambda)$ implies  $\lambda \leq \lambda_{MAX}(\kappa)$. This part of the proof is by contradiction. More specifically, we show that assuming $\lambda>\lambda_{MAX}(\kappa)$ contradicts the assumption $\kappa\leq \kappa_{MAX}(\lambda)$. Consider the set $E_{\lambda}$ that, as shown above, is not empty. Then, $\lambda>\lambda_{MAX}(\kappa)$ implies that
$\lambda>\lambda_{ij}(\kappa)$ for every $\{i,j\}\in E_{\lambda}$. Hence,  by (\ref{EQ.equiv.ineq.l.max}), for every $\{i,j\}\in E_{\lambda}$ it also holds that,
\begin{align*}
\kappa^2 + (s_{ii} + s_{jj})\kappa - g_{ij}(\lambda) > 0.
\end{align*}
In turn, it follows from (\ref{EQN:solution.kappa}) that  $\kappa> \kappa_{ij}(\lambda)$ for every $\{i,j\}\in E_{\lambda}$ and, thus, also that
$\kappa> \max_{\{i,j\}\in E_{\lambda}}\kappa_{ij}(\lambda)=\kappa_{MAX}(\lambda)$, and this completes the proof.
\section{Additional results for the simulation study}\label{AX:simulations.lasso}

This appendix contains additional results regarding the simulations presented in Section  \ref{SCT:simulations.lasso}.
Beside the quantities already considered in Figure~\ref{IMG.sim.lasso}, we also have computed the entropy loss, defined by
\[EL=\trace{\widehat{\Sigma}\Sigma^{-1}} + \log|\Sigma\widehat{\Sigma}^{-1}|,\]
as well as the True Positive Rate (TPR), True Negative Rate (TNR) and Positive Predicted Value (PPV).
Figure~\ref{IMG.sim.lasso.full} reports these results, confirming the main conclusions that were drawn in the main text.

\newcommand{\hspacing}{\hspace{0.015\textwidth}}
\newcommand{\subfwidth}{0.35\linewidth}
\begin{figure}[!htbp]
\vspace*{-15bp}
\centering
	\vspace{-2bp}
	\begin{subfigure}[t]{\subfwidth}
	\centering
	\caption{\label{IMG.sim.lasso.fulla}}
	\includegraphics[width = \linewidth]{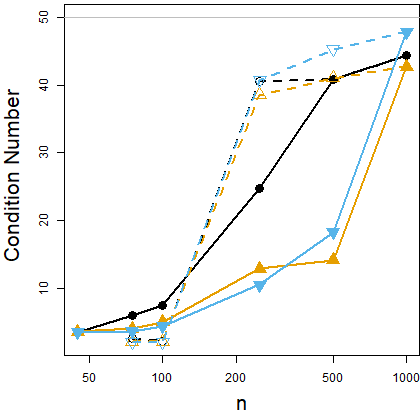}
	\end{subfigure} 
	\hspacing
	\vspace{-2bp}
	\begin{subfigure}[t]{\subfwidth}
	\centering
	\caption{\label{IMG.sim.lasso.fullb}}
	\includegraphics[width = \linewidth]{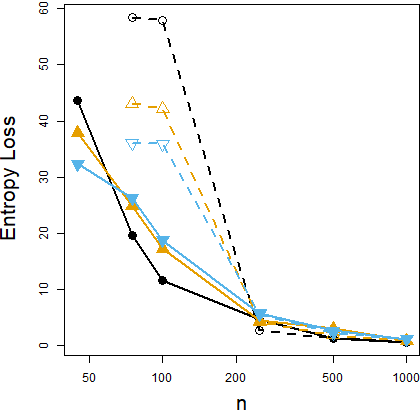}
	\end{subfigure} 
	\begin{subfigure}[t]{\subfwidth}
	\centering
	\caption{\label{IMG.sim.lasso.fullc}}
	\includegraphics[width = \linewidth]{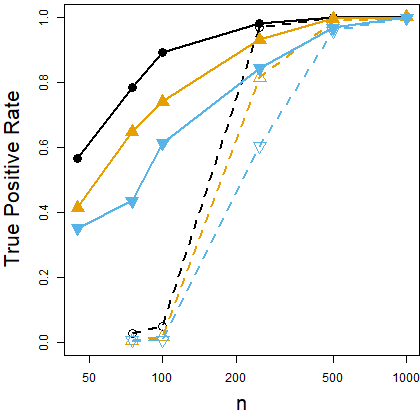}
	\end{subfigure} 
	\hspacing
	\vspace{-2bp}
	\begin{subfigure}[t]{\subfwidth}
	\centering
	\caption{\label{IMG.sim.lasso.fulld}}
	\includegraphics[width = \linewidth]{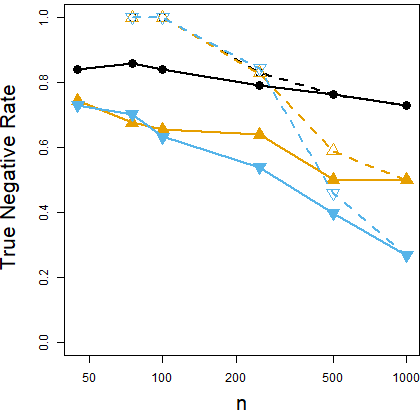}
	\end{subfigure} 
	\hspacing
	\vspace{-2bp}
	\begin{subfigure}[t]{\subfwidth}
	\centering
	\caption{\label{IMG.sim.lasso.fulle}}
	\includegraphics[width = \linewidth]{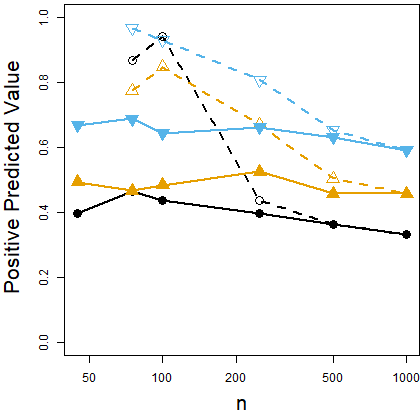}
	\end{subfigure} 
	\hspacing
	\begin{subfigure}[t]{\subfwidth}
	\centering
	\caption{\label{IMG.sim.lasso.fullf}}
	\includegraphics[width = \linewidth]{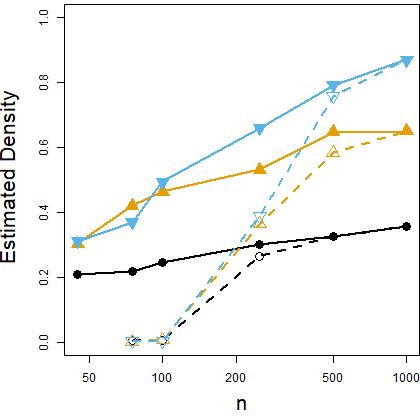}
	\end{subfigure} 

	\vspace{2bp}
	\includegraphics[scale=0.4]{legends/legend_lasso_horiz}
\caption{Average results from 10 simulations with sample size $n \in \{45, 75, 100, 250, 500, 1000\}$, $p=50$ and banded covariance matrices $\Sigma$ associated to graphs $\G$ with varying density levels (10\%, 30\%, 50\%). Comparison between \algoname{} ($\kappa \geq 0,\, \geq 0$) and covglasso ($\kappa = 0,\, \lambda \geq 0$) shows how  (a) Condition number of $\widehat{\Sigma}$ is underestimated in both cases but converges to the true value (horizontal line) for large $n$, (b) Entropy loss of $\widehat{\Sigma}$ with respect to $\Sigma$ is smaller for \algoname{}, i.e., a more accurate estimation, (c)-(d)-(e)-(f)  \algoname{} has a better recovery of the graph $\G$ than covglasso, which tends to select too sparse graphs for small $n$.
}\label{IMG.sim.lasso.full}
\end{figure}

\end{appendices}

\end{document}